\documentclass[10pt,twoside,leqno]{article}
\usepackage{amsfonts}
\usepackage{epsfig}
\usepackage{subfig}
\usepackage{amssymb,amsmath,bm}
\usepackage{amsthm}
\usepackage{mathrsfs}
\usepackage[section]{algorithm}
\usepackage{fancyhdr}
\theoremstyle{plain}
\newtheorem{proposition}{Proposition}

\newcounter{hypA}

\begin{document}

\bigskip

\begin{center}

{\Large \textbf{Inference for a Class of Partially Observed Point Process Models}}

\bigskip

BY JAMES S. MARTIN$^{1}$, AJAY JASRA$^{2}$ \& EMMA McCOY$^{1}$

{\footnotesize $^{1}$Department of Mathematics,
Imperial College London, London, SW7 2AZ, UK.}\\
{\footnotesize E-Mail:\,}\texttt{\emph{\footnotesize james.martin04@ic.ac.uk, e.mccoy@ic.ac.uk}}\\
{\footnotesize $^{2}$Department of Statistics \& Applied Probability,
National University of Singapore, Singapore, 117546, Sg.}\\
{\footnotesize E-Mail:\,}\texttt{\emph{\footnotesize staja@nus.edu.sg}}
\end{center}

\begin{abstract}
This paper presents a simulation-based framework for sequential
inference from partially and discretely observed point process
(PP's) models with static parameters. 
Taking on a Bayesian perspective for the static parameters,
we build upon sequential Monte Carlo (SMC) methods, investigating the
problems of performing sequential filtering and
smoothing in complex examples, where current methods often fail.
We consider various approaches
for approximating posterior distributions using SMC. Our approaches, with
some theoretical discussion
are illustrated on a doubly stochastic point process
applied in the context of
finance.\\
\begin{small}
\emph{Some Key Words}:  
Point Processes, Sequential Monte Carlo, Intensity Estimation
\end{small}
\end{abstract}

\section{Introduction}
\label{sec:intro}

Partially observed point processes provide a rich class of
models to describe real data. For example,
such models are used for stochastic
volatility (Barndorff-Nielsen \& Shephard, 2001) in finance, descriptions of queuing
data in operations research (Fearnhead, 2004), important seismological
models (Daley \& Vere-Jones, 1988) and applications in nuclear physics (Snyder
\& Miller, 1998).
For complex dynamic
models, that is, when data arrive sequentially in time, studies
date back to at least Snyder (1972). However, fitting Bayesian models
requires SMC (e.g.~Doucet et al.~(2000)) and Markov chain Monte Carlo (MCMC) methods. The
main developments in this field include the work of:
Centanni \& Minozzo (2006a,b); Green (1995); Del Moral et al.~(2006,2007); Doucet et al.~(2006);
Roberts et al.~(2004), Rydberg \& Shephard (2000), see also Whiteley et al.~(2011). As we describe below, the SMC methodology may fail in some scenarios
and we will describe methodology to deal with the problems that will be outlined.

Informally, the problem of interest is as follows. A process is
observed discretely upon a given time-interval $[0,T]$. The
objective is to draw inference at time-points
$t_0=0<t_1<\cdots<t_{\widetilde{m}}<T=t_{\widetilde{m}+1}$, on the
unobserved marked PP $(k_{t_n},\phi_{1:k_{t_n}},
\zeta_{1:k_{t_n}})$, where
$\phi_{1:k_{t_n}}=(\phi_1,\dots,\phi_{k_{t_n}})$ are the ordered
event times (constrained to $[0,t_n]$) and
$\zeta_{1:k_{t_n}}=(\zeta_1,\dots,\zeta_{k_{t_n}})$ are marks,
given the observations $y_{1:r_{t_n}}$. 
In other words to compute, for $n\geq 1$, at time $t_n$ 
\begin{align}
\pi_n(k_{t_n},\phi_{1:k_{t_n}},
\zeta_{1:k_{t_n}}|y_{1:r_{t_n}}) &~~
\textrm{smoothing}\label{eq:intro:filtering}
\\
\pi_n(k_{t_n}-k_{t_{n-1}},\phi_{k_{t_{n-1}}+1:k_{t_n}},
\zeta_{k_{t_{n-1}}+1:k_{t_n}}|y_{1:r_{t_n}})\label{eq:intro:smoothing}
 &~~ \textrm{filtering}.
\end{align}
In addition, there are static parameters 
specifying the probability model and these parameters will be
estimated in a Bayesian manner. At this stage a convention in our terminology
is established. An algorithm is said to be \emph{sequential} if it
is able to process data as it arrives over time. An algorithm is
said to be \emph{on-line} if it is sequential and has
a fixed computational cost per iteration/time-step. 

One of the first works applying computational methods to
PP models, was 
Rydberg \& Shephard (2000). They focus upon
a Cox model where the unobserved PP parameterizes the intensity of the observations. 
Rydberg \& Shephard (2000) used the auxiliary particle filter
(Pitt \& Shephard, 1997) to simulate from the posterior density of the
intensity at a given time point. This
was superseded by Centanni \& Minozzo (2006a,b), which
allows one to infer the intensity at any given time, up to the current
observation. Centanni \& Minozzo (2006a,b) perform an MCMC-type filtering
algorithm, estimating static
parameters using stochastic EM. The methodology cannot easily be adapted to the case where the
static parameters are given a prior distribution. 
In addition, the theoretical validity of the approach
has not been established; this is verified in Proposition \ref{prop:cmjustification} of this
article.

SMC samplers (Del Moral et al.~2006) are the focus of this paper and can be applied to all the problems
stated above. SMC methods simulate a set of $N\geq 1$
weighted samples, termed particles, in order to approximate a
sequence of distributions, which may be
chosen by the user, but which include (or closely related to)
the distributions in \eqref{eq:intro:filtering} and
\eqref{eq:intro:smoothing}. Such methods are provably convergent as $N\rightarrow\infty$ (Del Moral, 2004). A key feature of the approach is that the user must select:
\begin{enumerate}
\item{the sequence of distributions \label{enum:point1}}
\item{the mechanism by which particles are propagated \label{enum:point2}.}
\end{enumerate}
If points 1.~and 2.~are not properly addressed, there can be a substantial discrepancy
between the proposal and target; thus the variance of the
weights will be large and estimation 
inaccurate. This issue is particularly relevant when the targets
are defined on a sequence of nested spaces, as is the case for the
PP models -- the space of the point process trajectories becomes
larger with the time-parameter $n$. Thus, in choosing the sequence
of target distributions, we are faced with the question of how much the space should be enlarged at each iteration of the SMC
algorithm and how to choose a mechanism to propose
particles in the new region of the space.
This issue is referred to as \emph{the difficulty of extending the space}. 

Two solutions are proposed. The first is
to \emph{saturate} the state-space; it is supposed that the
observation interval, $[0,T]$, of the PP is known \emph{a priori}.
The sequence of target distributions is then defined on the whole interval and one
sequentially introduces likelihood terms. This idea circumvents the problem of
extending the space, at an extra computational cost. Inference for the original density of interest can be achieved by
importance sampling (IS). This approach 
cannot be used if $T$ is unkown. 
In the second approach, entitled \emph{data-point tempering}, the
sequence of target distributions are defined by sequentially
introducing likelihood terms.
This is achieved as follows: given that the PP has been sampled on $[0,t_n]$ the target is extended onto
$[0,t_{n+1}]$ by sampling the missing part of the PP. Then one introduces likelihood terms into the
target that correspond to the data (as in Chopin (2002)). Once all of the data have been introduced, the target
density is (\ref{eq:intro:filtering}). It should be noted that neither of the methods
are online, but some simple fixes are detailed.

Section \ref{sec:motivate} introduces a doubly stochastic PP model from finance
which serves as a running example. In Section \ref{sec:previous}
the ideas of Centanni \& Minozzo (2006a,b) are discussed; it is established that
the method is theoretically valid under some assumptions. The
difficulty of extending the state space is also demonstrated. In
Section \ref{sec:lpf} we introduce our SMC methods. In Section \ref{sec:financefinal} our
methods are illustrated on the running example. In Section
\ref{sec:summ} we detail extensions to our
work.

Some
notations are introduced. We consider a
sequence of probability measures $\{\varpi_n\}_{1\leq n \leq m^*}$ on spaces $\{(G_n,\mathcal{G}_n)\}_{1\leq n\leq m^*}$,
with dominating $\sigma-$finite measures.
 Bounded and measurable functions on $G_n$,
$f_n:G_n\rightarrow\mathbb{R}$, are written $\mathcal{B}_b(G_n)$ and
$\|f_n\|=\sup_{x\in G_n}|f_n(x)|$. $\varpi_n$ will refer to either the probability measure
$\varpi_n(dx)$ or density $\varpi_n(x)$.

\section{Model}\label{sec:motivate}

The model we use to illustrate our ideas is from statistical finance. An important
type of financial data is ultra high frequency
data which consists of the irregularly spaced times of financial
transactions and their corresponding monetary value.
Standard models for the fitting of such data have relied upon
stochastic differential equations driven by Wiener dynamics; a
debatable assumption due to the continuity of the sample paths. As
noted in Centanni \& Minozzo (2006b), it is more appropriate to model the data as a Cox process. Due to the high frequency of the
data, it is important to be able to perform sequential/on-line inference.
Data are observed in $[0,T]$. In the context of finance, the assumption
that $T$ be fixed is entirely reasonable. For example, when the model
is used in the context of equities, the model is run for the trading day; indeed due to different (deterministic) patterns in financial trading, it is likely that the fixed parameters below are varied according to the day.

A marked PP, of $r_T\geq 1$ points, is observed in 
time-period $[0,T]$. This is written
$y_{1:r_T}=(\omega_{1:r_T},\xi_{1:r_T})\in\Omega_{r,T}\times\Xi^{r_T}$
with
$\Omega_{r,T}=\{\omega_{1:r_T}:0<\omega_1<\cdots<\omega_{r_T}<T\}$,
$\Xi\subseteq\mathbb{R}$. Here the $\omega$ are the transaction
times and $\xi$ are the log-returns on the financial transactions.
An appropriate model for such data, as in Centanni \& Minozzo (2006b), is
\begin{eqnarray*}
\tilde{p}(\xi_{1:r_T}|\mu,\sigma) & = & \prod_{i=1}^{r_T} \tilde{p}(\xi_i;\mu,\sigma)\\
\tilde{p}(\omega_{1:r_T}|\{\lambda_{T}\}) & \propto &
\prod_{i=1}^{r_T}\big\{\lambda_{\omega_i}\big\}\exp\bigg\{-\int_{0}^T\lambda_u
du\bigg\}
\end{eqnarray*}
with $\tilde{p}$ a generic density,  $\xi_i|\mu,\sigma$ are assumed to be $t$-distributed
on 1 degree of freedom, location $\mu$, scale $\sigma$ and
$\lambda_u$ is the intensity. 
The unobserved intensity process is assumed to follow
the dynamics
$d\lambda_t  =  -s \lambda_t dt + dJ_t$
with $\{J_t\}$ a compound Poisson process:
$
J_t = \sum_{j=1}^{k_t} \zeta_j
$
with $\{K_t\}$ a Poisson process with rate parameter $\nu$ and
i.i.d. jumps $\zeta_j\sim\mathcal{E}x(1/\gamma)$,
$\mathcal{E}x(\cdot)$ is the exponential distribution. That is,
for $t\in [0,T]$,
\begin{equation}
\lambda_t = \bigg\{\lambda_0e^{-st} +
\sum_{j=1}^{k_t}\zeta_j
e^{-s(t-\phi_j)}\bigg\}\label{eq:intensity_def}
\end{equation}
with $\phi_j$ the jump times of the unobserved Poisson process and $\lambda_0$ fixed throughout (using a short preliminary time series that is available in practice).

We define the following notation:
\begin{align*}
\bar{x}_n &= (k_{t_n}, \phi_{1:k_{t_n}}, \zeta_{1:k_{t_n}}),\\
\bar{x}_{n,1}&=(k_{t_{n}}-k_{t_{n-1}}, \phi_{k_{t_{n-1}}+1:k_{t_n}}, \zeta_{k_{t_{n-1}}+1:k_{t_n}}),\\
\bar{y}_n &= (\omega_{1:r_{t_n}}, \xi_{1:r_{t_n}}),\\
\bar{y}_{n,1}& = (\omega_{r_{t_{n-1}}+1:r_{t_n}},
\xi_{r_{t_{n-1}}+1:r_{t_n}}).
\end{align*}
Here $\bar{x}_n$ (respectively $\bar{y}_n$) is
the the restriction of the hidden (observed) PP to 
events in $[0,t_n]$. Similarly $\bar{x}_{n,1}$ (respectively
$\bar{y}_{n,1}$) is the the restriction of the hidden
(observed) PP to events in $[t_{n-1}, t_n]$.

The objective is to perform inference at times
$0<t_1<\cdots<t_{\widetilde{m}}<T=t_{\widetilde{m}+1}$, that is,
to update the posterior distribution conditional on the data
arriving in $[t_{n-1},t_n]$. To summarize, the posterior
distribution at time $t_n$ is
$$
\pi_n(\bar{x}_n,\mu,\sigma|\bar{y}_n)
\propto
\prod_{i=1}^{r_{t_n}}\big\{\tilde{p}(\xi_i;\mu,\sigma)\lambda_{\omega_i}\big\}\exp\bigg\{-\int_{0}^{t_n}\lambda_u
du\bigg\}\times \prod_{i=1}^{k_{t_n}}\big\{\mathsf{p}(\zeta_i)\big\}
\mathsf{p}(\phi_{1:k_{t_n}})\mathsf{p}(k_{t_n})\times \tilde{p}(\mu,\sigma)
$$
\begin{equation}
= l_{[0,{t_n}]}(\bar{y}_n;\bar{x}_n,\mu,\sigma)\times\\
\mathsf{p}(\bar{x}_n)\times
\tilde{p}(\mu,\sigma)\label{eq:posterior}
\end{equation}
with $l_{[0,{t_n}]}$ corresponding to the first part of the equation above,
$ \mu \sim \mathcal{N}(\alpha_{\mu},\beta_{\mu})$, 
$\sigma \sim \mathcal{G}a(\alpha_{\sigma},\beta_{\sigma})$,
$\phi_{1:k_t}|k_t \sim \mathcal{U}_{\Phi_{k,t_n}}$, 
$k_t \sim \mathcal{P}o(\gamma t)$ and
where $\mathcal{U}_A$ is the uniform distribution on the set $A$,
$\mathcal{N}(\mu,\sigma^2)$ is the normal distribution of mean
$\mu$ and variance $\sigma^2$, $\mathcal{G}a(\alpha,\beta)$ the Gamma
distribution of mean $\alpha/\beta$ and $\mathcal{P}o$ is the Poisson
distribution. $\mathsf{p}(\bar{x}_n)$ is the notation for the prior on the marked point-process and $\tilde{p}(\mu,\sigma)$
is the notation for the prior on $(\mu,\sigma)$.
Later a $\pi_0$ is introduced which will refer to an initial
distribution. Note it is possible to perform inference on
$(\mu,\sigma)$ independently of the unobserved PP; it will not
significantly complicate the simulation methods to include them.

It is of interest to compute expectations w.r.t.~the
$\{\pi_n\}_{1\leq n\leq m^*}$, and this is possible, using the SMC
methods below (Section \ref{sec:smcmethods}). However, such
algorithms are not of fixed computational cost; the sequence of
spaces over which the $\{\pi_n\}_{1\leq n\leq m^*}$ lie is increasing.
These methods can also be used to draw inference from the marginal posterior of the process, over $(t_{n-1},t_n]$; such algorithms can be designed to be of fixed computational complexity, for example by constraining any simulation to a fixed-size state-space. This idea is considered further in Section \ref{sec:online_disc}.



\section{Previous Approaches}\label{sec:previous}


%

One of the approaches for performing filtering for partially
observed PP's is from Centanni \& Minozzo (2006a). In this Section the parameters
$(\mu,\sigma)$ are assumed known. Let
$$
\bar{E}_n =
\bigcup_{k\in\mathbb{N}_0}\bigg(\{k\}\times\Phi_{k,t_n}\times(\mathbb{R}^+)^{k}\bigg).
$$
This is the support of the target densities for this method.



The following decomposition is adopted
\begin{eqnarray}
\pi_n(\bar{x}_n|\bar{y}_n) & = & \frac{l_{(t_{n-1},
t_n]}(\bar{y}_{n,1};\bar{x}_n) }{p_n(\bar{y}_{n,1}|\bar{y}_{n-1})}
\mathsf{p}(\bar{x}_{n,1})
\pi_{n-1}(\bar{x}_{n-1}|\bar{y}_{n-1})\\
\label{eq:centrec} \tilde{p}_n(\bar{y}_{n,1}|\bar{y}_{n-1}) & = & \int
l_{(t_{n-1}, t_n]}(\bar{y}_{n,1};\bar{x}_n)
\mathsf{p}(\bar{x}_{n,1})
\pi_{n-1}(\bar{x}_{n-1}|\bar{y}_{n-1})d\bar{x}_n. \nonumber
\end{eqnarray}
At time $n\geq 2$ of the algorithm, a reversible jump MCMC kernel
(although the analysis below is not restricted to such scenarios)
is used for $N$ steps to sample from the approximated density
\begin{eqnarray*}
\pi_n^N(\bar{x}_n|\bar{y}_n) & \propto & l_{(t_{n-1},
t_n]}(\bar{y}_{n,1};\bar{x}_n) \mathsf{p}(\bar{x}_{n,1}) S_{x,
n-1}^{N}(\bar{x}_{n-1})
\end{eqnarray*}
where
$S_{x, n-1}^N(\bar{x}_{n-1}) := \frac{1}{N}\sum_{i=1}^N\mathbb{I}_{\{\bar{X}_{n-1}^{(i)}\}}(\bar{x}_{n-1})$
with $\bar{X}_{n-1}^{(1)},\dots,\bar{X}_{n-1}^{(N)}$ obtained from
a reversible jump MCMC algorithm of invariant measure
$\pi_{n-1}^N$. The algorithm for $n=1$ targets $\pi_1$ exactly;
there is no empirical density $S_{x, 0}^N$. At time $n=1$ the
algorithm starts from an arbitrary point 
$\bar{x}_1^{(1)}\in\bar{E}_1$ and subsequent steps are initialized
by a draw from the empirical $S_{x,n-1}^N$ and the prior
$\mathsf{p}$ (this can be modified); $N-1$ additional samples are
simulated.

The above algorithm can be justified, theoretically, by using the
Poisson equation (e.g.~Glynn \& Meyn (1996)) and induction arguments. Below
the assumption (A) is made; see the appendix for the assumption (A) as well
as the proof. Also, the expectation below is w.r.t.~the process
discussed above, given the observed data.

\begin{proposition}\label{prop:cmjustification}
Assume (A). Then for any $n\geq 1$, $\bar{y}_n$, $p\geq 1$ there
exists $B_{p,n}(\bar{y}_n)<+\infty$ such that for any
$f_n\in\mathcal{B}_b(\bar{E}_n)$ \begin{equation}
\mathbb{E}_{\bar{x}_{1}^{(1)}}\bigg[\bigg|\frac{1}{N}\sum_{i=1}^Nf_n(\bar{X}_n^{(i)})-\int_{\bar{E}_n}f_n(\bar{x}_n)\pi_n(d\bar{x}_{n})\bigg|^p
\bigg|\bar{y}_n\bigg]^{1/p}
\leq \frac{B_{p,n}(\bar{y}_n)\|f_n\|}{\sqrt{N}}.
\label{eq:lpbound}
\end{equation}
\end{proposition}
This result helps to establish the theoretical validity of the
method in Centanni \& Minozzo (2006a), which to our knowledge, had not been
established in that paper or elsewhere. In addition, it allows us
to understand where and when the method may be of use; this is
discussed in Section \ref{sec:smccm_comparison}.

\subsection{SMC Methods}\label{sec:smcmethods}

SMC samplers aim to approximate a sequence of related probability measures
$\{\pi_n\}_{0\leq n \leq m^*}$ defined upon a common space
$(E,\mathcal{E})$. Note that $m^*>1$ can depend upon the data and
may not be known prior to simulation. For partially observed PPs the
probability measures are defined upon nested state-spaces: this
case can be similarly handled with minor modification.
SMC samplers introduces a sequence of
auxiliary probability measures $\{\widetilde{\pi}_n\}_{0\leq n
\leq m^*}$ on state-spaces of increasing dimension
$(E_{[0,n]}:=E_0\times\cdots\times
E_n,\mathcal{E}_{[0,n]}:=\mathcal{E}_0\otimes\cdots\otimes\mathcal{E}_n)$,
such that they admit the $\{\pi_n\}_{0\leq n\leq m^*}$ as
marginals.

The following sequence of auxiliary densities is used:
\begin{equation}
\widetilde{\pi}_n(x_{0:n}) = \pi_n(x_n)\prod_{j=0}^{n-1}L_{j}(x_{j+1},x_j)
\label{eq:smcsampaux}
\end{equation}
where $\{L_{n}\}_{0\leq n \leq m^*-1}$ are
backward Markov kernels. In our application $\pi_0$ is the
prior, on $E_1$ (as defined below). It is clear that
(\ref{eq:smcsampaux}) admit the $\{\pi_n\}$ as marginals, and
hence these distributions can be targeted using precisely the same
mechanism as in sequential importance sampling/resampling; the algorithm is given in Figure
\ref{fig:smcsampleralgo}. 

The ESS in Figure \ref{fig:smcsampleralgo} refers to the effective
sample size (Liu, 2001). This measures the weight degeneracy of
the algorithm; if the ESS is close to $N$, then this indicates
that all of the samples are approximately independent. This is a
standard metric by which to assess the performance the algorithm.
The resampling method used throughout the paper is systematic
resampling.

One generic approach is to set $K_n$ as an MCMC kernel of invariant distribution
$\pi_n$ and $L_{n-1}$ as the reversal kernel
$
L_{n-1}(x_n,x_{n-1}) =
\pi_n(x_{n-1})K_n(x_{n-1},x_n)/\pi_n(x_n)
$
which we term the \emph{standard reversal kernel}.
One can iterate the MCMC kernels, by which we use the positive integer $M$ to denote the number of iterates.
It is also possible to apply the algorithm when $K_n$ is a mixture of kernels;
see Del Moral et al.~(2006) for details on the algorithm.

\begin{figure}[h]
\begin{itemize}
\item{0. Set $n=0$; for $i=1,\dots,N$ sample $X_0^{(i)}\sim\varrho_0$ and compute $w_0(X^{(i)}_0)\propto\pi_0(X^{(i)}_0)/\varrho_0(X^{(i)}_0)$. }
\item{1.Compute the normalized importance weights,
$$
w_n^{(i)} = \frac{w_{n}(X_{0:n}^{(i)})}{\sum_{j=1}^N
w_{n}(X_{0:n}^{(j)})}
$$
if the
$\textrm{ESS}=\big\{\sum_{j=1}^Nw_n(X_{0:n}^{(j)})\big\}^2/\sum_{j=1}^N\{w_n(X_{0:n}^{(j)})\}^2<T(N)$
then resample the particles and set the importance weights to
uniform. Set $n=n+1$, if $n=m^*+1$ stop.}
\item{2. For $i=1,\dots,N$ sample $X_n^{(i)}|X_{n-1}^{(i)}=x_{n-1}^{(i)}\sim K_n(x_{n-1}^{(i)},\cdot)$,
and compute:
\begin{equation}
W_n(X_{n-1:n}^{(i)}) \propto
\frac{\pi_n(X_n^{(i)})L_{n-1}(X_{n-1}^{(i)},X_n^{(i)})}
{\pi_{n-1}(X_{n-1}^{(i)})K_n(X_{n-1}^{(i)},X_n^{(i)})}\label{eq:incrementalweight}
\end{equation}
$w_{n}(X^{(i)}_{0:n})=W_n(X^{(i)}_{n-1:n})w_{n-1}(X_{0:n-1}^{(i)})$
and return to the start of 1.}
\end{itemize}
\caption{A Generic SMC Sampler. Note that $T(N)$ is termed a threshold function
such that $1\leq T(N) \leq N$ and ESS is the effective sample size.}
\label{fig:smcsampleralgo}
\end{figure}

\subsubsection{Nested Spaces}\label{sec:nested:spaces}
As described in Section \ref{sec:intro}, in complex
problems it is often difficult to design efficient SMC algorithms.
In the example in Section \ref{sec:motivate}, the state-spaces of
the subsequent densities are not common. The objective is to
sample from a sequence of densities on the space, at time $n$,
$$
E_n = \bigg(\bigcup_{k\in\mathbb{N}_0}\{k\}\times\Phi_{k,t_n}\times(\mathbb{R}^+)^{k}\bigg)
\times\mathbb{R}\times\mathbb{R}^+\quad 1\leq n\leq m^*-1 
$$
with $E_0=E_1$.
That is, for any $1\leq n\leq m^*-1$, $E_{n}\subseteq E_{n+1}$.
Two standard methods for extending the space, as in Del Moral et al.~(2006)
are to propagate particles by application of  `birth' and the
`extend' moves.

Consider the model in Section \ref{sec:motivate}. The following SMC steps are used to extend the space at time $n$ of the algorithm.
\begin{itemize}
\item{\textbf{Birth}. A new jump is sampled
uniformly in $[\phi_{k_{t_{n-1}}},t_n]$ and a new mark from the prior.
The incremental weight is
$$
W_n(\bar{x}_{n-1:n},\mu,\sigma) \propto
\frac{\pi_n(\bar{x}_n,\mu,\sigma|\bar{y}_n)(t_n-\phi_{k_{t_{n-1}}})}
{\pi_{n-1}(\bar{x}_{n-1},\mu,\sigma|\bar{y}_n)\mathsf{p}(\zeta_{k_{t_n}})}.
$$
}
\item{\textbf{Extend}. A new jump is generated
according to a Markov kernel that corresponds to the random walk:
$$
\log\bigg\{\frac{\phi_{k_{t_n}} - \phi_{k_{t_n}-1}}{t_n-\phi_{k_{t_n}}}\bigg\} =
\vartheta Z + \log\bigg\{\frac{\phi_{k_{t_{n-1}}} - \phi_{k_{t_{n-1}}-1}}{t_n-\phi_{k_{t_{n-1}}}}\bigg\}
$$
with $Z\sim\mathcal{N}(0,1)$, $\vartheta>0$. The new mark is sampled from the prior.
The backward kernel and incremental weight are discussed in Del Moral et al.~(2007) section 4.3.}
\end{itemize}
Note, as remarked in Whiteley et al.~(2011), we need to be able to sample
any number of births. With an
extremely small probability, a proposal from the prior is
included to form a mixture kernel.

In addition to the above steps an MCMC sweep is included after the decision of whether or not to resample
the particles is taken (see step 1.~of Figure \ref{fig:smcsampleralgo}): an MCMC kernel of invariant measure $\pi_n$ is applied.
The kernel is much the same as in Green (1995).

\subsubsection{Simulation Experiment}\label{sec:simosfinance1}

We applied the benchmark sampler, as detailed above, to some synthetic data in order to monitor the performance of the algorithm. Standard practice in the reporting of financial data is to represent the time of a trade as a positive real number, with the integer part representing the number of days passed since January $1^{\mbox{st}}$ 1900 and the non-integer part representing the fraction of 24 hours that has passed during that day; thus, one minute corresponds to an interval of length $1/1440$. Therefore we use a synthetic data set with intensity of order of magnitude $10^3$. The ticks $\omega_i$ were generated from a specified intensity process $\left\{\lambda_t\right\}$ that varied smoothly between three levels of constant intensity at $\lambda=6000$, $\lambda=2000$ and $\lambda=4000$. The log returns $\xi_i$ were sampled from the Cauchy-distribution, location $\mu=0$ and scale $\sigma=2.5\times10^{-4}$. The entire data set was of size $r_T=3206$, $[0,T]=[0,0.9]$ with $t_n= n*0.003$. The intensity from which they were generated had constant levels at 6000 in the interval [0.05,0.18]; at 4000 in the interval [0.51,0.68]; and at 2000 in the intervals [0.28,0.42] and [0.78,0.90].

The sampler was implemented with all permutations $\left\{\left(M,N\right)\right\}$ for $N\in\left\{100,1000\right\}$ and $M\in\left\{1,5,20\right\}$, resampling whenever the effective sample size fell below $N/2$ (recall $N$ is the number of particles and $M$ the MCMC iterations). 
When performing statistical inference,
the intensity \eqref{eq:intensity_def} used parameters $\gamma=0.001$, $\nu=150$ and $s=20$.

It was found that for this SMC sampler, the system consistently collapses to a single particle representation of the distribution of interest within an extremely short time period.  That is, resampling is needed at almost every time step, which leads to an extremely poor representation of the target density. Figure \ref{ESS1} shows the ESS at each time step for a paricular implementation. As can be seen, the algorithm behaves extremely poorly for this model.

\begin{figure}[h!]
\begin{center}
\includegraphics[width=\textwidth,height=8cm]{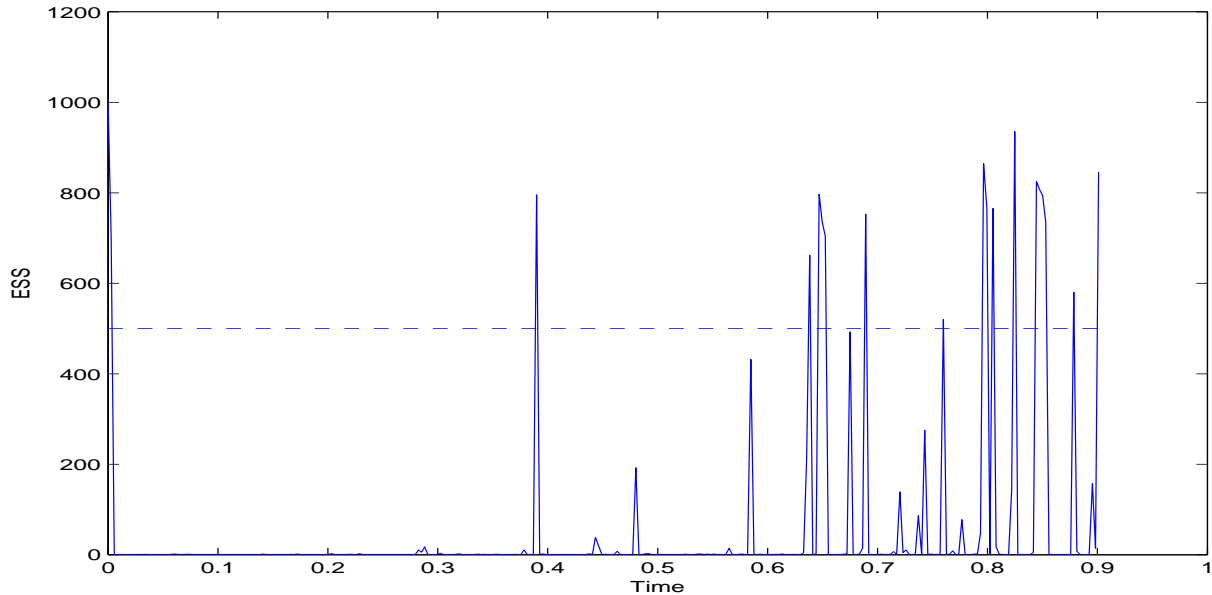}
\caption{Effective Sample Size plots for the SMC sampler described in Figure \ref{fig:smcsampleralgo}, implemented with $N=1000$ particles and with $M=5$ MCMC sweeps at each iteration. The dashed line indicates the resampling threshold at $N/2=500$ particles; resampling is needed at 94.4\% of the time steps.}
\label{ESS1}
\end{center}
\end{figure}

\subsection{Discussion}\label{sec:smccm_comparison}

We have reviewed two existing techniques for the Bayesian analysis
of partially observed PP's. It should be noted that there are
other methods, for example in Varini (2007). In that paper, the
intensity has a finite number of functional forms and the
uncertainty is related to the type of form at each inference time
$t_n$.

The relative advantage of the approach of Centanni \& Minozzo (2006a) is
the fact that the state-space need not be extended.
On page 1586 of Centanni \& Minozzo (2006a) the authors describe the filtering/smoothing algorithm, for
the process on the entire interval $[0,t_n]$ at time $n$; the theory discussed
in Proposition \ref{prop:cmjustification} suggests that this method is not
likely to work well as $n$ grows. The bound, which is perhaps a little loose
is, for $n\geq 2$
$$
B_{p,n}(\bar{y}_n) = \frac{2}{\epsilon_n(\bar{y}_n)}[B_p + 1] +
\hat{k}_n B_{p,n-1}(\bar{y}_{n-1})
$$
 with
$B_{p,1}(\bar{y}_1)=\frac{2}{\epsilon_1(\bar{y}_1)}[B_p + 1]$,
$B_p$ a constant related to the B\"urkholder/Davis inequalities
(e.g.~Shiryaev (1996)), $\epsilon_n(\bar{y}_n)\in(0,1)$ and
$\hat{k}_n>0$ a constant that is model/data dependent which is
possibly bigger than 1. The bound indicates that the error can increase over time, even under the exceptionally strong
assumption (A) in the appendix. This is opposed to SMC methods
which are provably stable, under similar assumptions (and that the entire
state is updated), as
$n\rightarrow\infty$ (Del Moral, 2004). In other words, whilst the
approach of Centanni \& Minozzo is useful in difficult problems,
it is less general with potentially slower convergence rate than SMC.
Intuitively, it seems that the method of Centanni \& Minozzo (2006a) is perhaps
only useful when considering the process on
$(t_{n-1},t_n]$, as the process further back in time is not
rejuvenated in any way.  As a result,
parameter estimation may not be very accurate. In addition, the
method cannot be extended to a sequential algorithm such that
fully Bayesian inference is possible. As noted above, SMC samplers
can be used in such contexts, but requires a computational budget
that grows with the time parameter $n$.

As mentioned above, SMC methods are provably stable under some conditions as the time parameter grows. However, some remarks related to the method in Figure \ref{fig:smcsampleralgo} can help to shed some light on the poor behaviour in Section \ref{sec:simosfinance1}. Consider the scenario when one is interested in statistical inference on $[0,t_1]$. Suppose for simplicity, one can write the posterior on this region as
\begin{equation}
\pi(\bar{x}_1) \propto \exp\{\sum_{i=1}^{r_1} g_i(y_i;\bar{x}_1)\}\mathsf{p}(\bar{x}_1)\label{eq:stability_disc}
\end{equation}
for fixed $r_1,\mu,\sigma$. 
If one considers just pure importance sampling, then conditioning upon the data, one can easily  show that for any $\pi_1-$(square) integrable $f$
with $\int f(\bar{x}_1)\mathsf{p}(\bar{x}_1)d\bar{x}_1=0$, the asymptotic variance in the associated Gaussian central limit theorem
is lower-bounded by:
$$
\Big(\int f(\bar{x}_1)^2\exp\{2\sum_{i=1}^{r_1} g_i(y_i;\bar{x}_1)\} \mathsf{p}(\bar{x}_1) d\bar{x}_1\Big)/
\Big(\int \exp\{\sum_{i=1}^{r_1} g_i(y_i;\bar{x}_1)\} \mathsf{p}(\bar{x}_1) d\bar{x}_1\Big).
$$
Then, for any mixing type sequence of data the asymptotic variance will for some $f$ and in \emph{some} scenarios, grow without bound as $r_1$ grows - this is a very heuristic observation, that requires further investigation.
Hence, given this discussion and our empirical experience, it seems that we require a new methodology, especially for complex problems.

\subsection{Possible Solutions to the problems of Extending the State-Space}

An important remark associated to the simulations in Section \ref{sec:simosfinance1}, is that it cannot be
expected that simply increasing the number of particles will necessarily a significantly better estimation procedure.
The algorithm completely crashes to a single particle and it seems
that naively increasing computation will not improve the simulations.

As discussed above, the inherent difficulty of sampling from the given sequence of distributions
is that of extending the state-space. It is known that conditional on all parameters except the final jump, the optimal importance distribution
is the full conditional density (Del Moral et al.~2006). In practice, for many problems it is either not possible to sample
from this density, or to evaluate it exactly (which is required).
In the case that it is
possible to sample from the full conditional, but the normalizing
constant is unknown, the normalizing constant problem can be dealt with via
the random weight idea (Rousset \& Doucet, 2006). In the context of this
problem we found that the simulation from the full conditional density of
$\phi_{k_{t_n}}$ was difficult, to the extent that sensible rejection algorithms
and approximations for the random weight technique were extremely poor.

Another solution, in Del Moral et al.~(2007), consists of stopping the algorithm
when the effective sample size (ESS) drops and using an additional SMC sampler
to facilitate the extension of the state-space. However, in this example,
the ESS is so low, that it cannot be expected to help.
Due to above discussion, it is clear that a new technique is required to
sample from the sequence of distributions; two ideas are presented below.
One idea, in the context of estimating static parameters, that could be adopted is
SMC$^2$ (Chopin et al.~2011) which has appeared after the first versions of this article.

\section{Proposed Methods}\label{sec:lpf}

In the following Section, two approaches are presented to deal
with the problems in Section \ref{sec:simosfinance1}. First,
a state-space saturation approach, where sampling of PP trajectories is performed over a state space corresponding to a fixed observation interval. Second, a data-point tempering approach. In this approach, as the time parameter increases, the (artificial)
target in the new region is simply the prior and the data are then
sequentially added to the likelihood, softening the state-space
extension problem. Both of these procedures use the basic structure of Figure \ref{fig:smcsampleralgo}, with some refinements, that are mentioned in the text. As for the algorithms in Figure \ref{fig:smcsampleralgo} we add dynamic resampling steps; when  MCMC kernels
are used, one can resample before sampling - see Del Moral et al.~(2006) for details.

\subsection{Saturating the
State-Space}\label{sec:saturating_the_space}

A simple idea, which has been used in the context of reversible
jump, is to saturate the state-space. The idea relies upon knowing
the observation period of the PP ($[0,T]$) \emph{a priori} to
the simulation. This is realistic in a variety of
applications. For example, in Section \ref{sec:motivate}, often we
may only be interested in performing inference for a day of
trading and thus can set $[0,T]$.

In details, it is proposed to sample, in the case of the example in Section
\ref{sec:motivate}, from the sequence of target densities defined on the
space
\begin{equation}
E =
\bigg(\bigcup_{k\in\mathbb{N}_0}\{k\}\times\Phi_{k,T}\times(\mathbb{R}^+)^{k}\bigg)
\times\mathbb{R}\times(\mathbb{R}^+)^2.\label{eq:saturated_space}
\end{equation}
The (marginal, that is in the sense of \eqref{eq:smcsampaux}) target densities are now, denoted with a $S$ as a super-script:
$$
\pi_n^S(\bar{x}_n,\mu,\sigma|\bar{y}_n) \propto 
\prod_{i=1}^{r_{t_n}}\big\{\tilde{p}(\xi_i;\mu,\sigma)\lambda_{\omega_i}\big\}\exp\bigg\{-\int_{0}^{t_n}\lambda_u
du\bigg\}\times \prod_{i=1}^{k_{t_n}}\big\{\mathsf{p}(\zeta_i)\big\}
\mathsf{p}^S(\phi_{1:k_{t_n}})\mathsf{p}^S(k_{t_n})\times \tilde{p}(\mu,\sigma)\quad 1\leq n \leq T
$$
where the prior on the point process is $\phi_{1:k_t}|k_t \sim \mathcal{U}_{\Phi_{k,T}}$, 
$k_t \sim \mathcal{P}o(\gamma T)$.
We then use, for
$K_n$, an MCMC kernel of invariant measure $\pi_n^S$ and the standard reversal kernel
discussed in Section \ref{sec:smcmethods} for the backward kernel. 
The initial distribution is the prior and the weight at time 0 is proportional to 1 for each particle.
The incremental weights at subsequent time-points are simply: 
$$
W_n(\bar{x}_{n-1},\mu_{n-1},\sigma_{n-1}) \propto 
\frac{\pi_n^S(\bar{x}_{n-1},\mu_{n-1},\sigma_{n-1}|\bar{y}_n)}{\pi_{n-1}^S(\bar{x}_{n-1},\mu_{n-1},\sigma_{n-1}|\bar{y}_{n-1})}
\quad 1\leq n \leq T.
$$
Inference
w.r.t.~the original $\{\pi_n\}_{1\leq n \leq m^*}$ can  be
performed via IS as the supports of the targets of interest are contained within the proposals (i.e.~via the targets of the saturated algorithm).

\subsection{Data-Point Tempering}

 A simple solution to the state-space extension problem, which allows data to
be incorporated sequentially, albeit not being of fixed
computational complexity is as follows. When the time parameter
increases, the new part of the process is simulated according to
the prior. Then each new data point is added to the likelihood in
a sequential manner. In other words if there are $n$ data points, then
there are $m^* = n + \widetilde{m}$ time-steps of the algorithm.

To illustrate, consider only the scenario of the data in $[0,t_1]$, with $r_{t_1}>0$.
Then our sequence of (marginal) targets are: $\pi_0^{\textrm{TE}}(\bar{x}_1,\mu,\sigma)=\mathsf{p}(\bar{x}_1)\tilde{p}(\mu)\tilde{p}(\sigma)$
and for $1\leq n \leq r_{t_1}$
$$
\pi_n^{\textrm{TE}}(\bar{x}_1,\mu,\sigma|y_{1:n}) \propto 
\prod_{i=1}^{n}\big\{\tilde{p}(\xi_i;\mu,\sigma)\lambda_{\omega_i}\big\}\exp\bigg\{-\int_{0}^{t_1}\lambda_u
du\bigg\}
\mathsf{p}(\bar{x}_1)\tilde{p}(\mu)\tilde{p}(\sigma).
$$
Then, when considering the extension of the point-process onto $[0,t_2]$, one has a (marginal) target that is:
$$
\pi_{r_{t_1}+1}^{\textrm{TE}}(\bar{x}_2,\mu,\sigma|\bar{y}_1) \propto \prod_{i=1}^{r_{t_1}}\big\{\tilde{p}(\xi_i;\mu,\sigma)\lambda_{\omega_i}\big\}\exp\bigg\{-\int_{0}^{t_1}\lambda_u
du\bigg\}
\mathsf{p}(\bar{x}_2)\tilde{p}(\mu)\tilde{p}(\sigma)
$$
When one extends the state-space, we sample from the prior on the new segment, which leads to a unit incremental weight (up-to proportionality) - no backward kernel is required here.
Then, when adding data, we simply use MCMC kernels to move the particles (the kernels as in Section \ref{sec:nested:spaces}) and the standard reversal kernel
discussed in Section \ref{sec:smcmethods} for the backward kernel.
This leads to an incremental weight that is the ratio of the consecutive densities at the previous state. 

The potential advantage of this idea is that, when extending the state-space,
there is no extra data, to potentially complicate the likelihood.
Thus, it is expected that if the prior does not propose a significant number
of new jumps, that the incremental weights should be of relatively low variance.
The subsequent steps, when considering the jumps in $[t_n,t_{n+1})$ are performed
on a common state-space and hence should not be subject to as substantial
variability as when the state-space changes. This idea could also be adapted
to the case that the likelihood on the new interval are tempered instead (e.g.~Jasra et al.~(2007)).

As a theoretical investigation of this idea, we return to the discussion of Section \ref{sec:smccm_comparison} and in particular, where the joint target density is \eqref{eq:stability_disc}. 
We consider the data-point tempering which starts with a draw from the prior and sequentially adds data points. In otherwords runs for $r_1+1$ time-steps with
$$
\pi_n(\bar{x}_1) \propto \exp\{\sum_{i=1}^n g_i(y_i;\bar{x}_1)\}\mathsf{p}(\bar{x}_1)
$$
with a $-\infty<\underline{g}<\overline{g}<\infty$ such that for each $i$, $y_i$ and all $\bar{x}_1$,  $\underline{g} \leq  g_i(y_i;\bar{x}_1) \leq \overline{g}$.
The algorithm resamples at every time-step and uses MCMC kernels, which are assumed to satisfy, for some $\tau\in(0,1)$,
and each $1\leq n \leq r_1$, $r_1$, $,\bar{x}_1,\bar{x}_1'$
$$
K_n(\bar{x}_1,\cdot)\geq \tau K_n(\bar{x}_1',\cdot).
$$
At the very final time-step one also resamples after the final weighting of the particles.
Write $\bar{X}_1^1,\dots,\bar{X}_1^N$ as the samples that approximate target \eqref{eq:stability_disc}. 
Suppose $f\in\mathcal{B}_b(\bar{E}_1)$, then
there is a 
Gaussian central limit theorem for
$$
\sqrt{N}\Big(\frac{1}{N}\sum_{i=1}^N f(\bar{X}_1^i) -  \int_{\bar{E}_1}f(\bar{x}_1)\pi_{r_1}(\bar{x}_1)d\bar{x}_1\Big).
$$
 Writing the asymptotic variance as $\sigma^2_{\textrm{TE},r_1}(f)$,
we have the following result whose proof is in the appendix.

\begin{proposition}\label{prop:data_point}
For SMC sampler described above, with final target \eqref{eq:stability_disc} then we have for any $f\in\mathcal{B}_b(\bar{E}_1)$
that there exists a $B\in(0,+\infty)$ such that for any 
$r_1\geq 1$, $\bar{y_1}$ 
$$
\sigma^2_{\textrm{TE},r_1}(f) \leq B.
$$
\end{proposition}

The upper-bound does not grow with the number of data. That is, by increasing the computational complexity linearly in the number of data, one has an algorithm whose error does not grow as more data (and regions) are added. 
This is similar to the observation of Beskos et al.~(2011), when increasing the dimension of the target density.
We note that the result is derived under exceptionally strong assumptions. In general, when one considers $r_1$ growing, one requires sharper tools than the Dobrushin coefficients used here (e.g.~Eberle \& Marinelli (2011)); this is beyond the scope of the current article and our
result above is illustrative (and hence potentially over-optimistic).

\subsection{Online Implementation}\label{sec:online_disc}
A key characteristic that has not yet been addressed is the fact that each has a computational complexity that is increasing with time. In a procedure that would otherwise be well suited to providing online inference, this is an unattractive feature. A large contribution to this increasing computational budget derives from the MCMC sweeps at the end of each iteration. As the space over which the invariant MCMC kernel is being applied is increased, so does the expense of the algorithm. An improvement to the computational demand of the samplers can therefore be made by keeping the space over which the MCMC kernel is applied constant. The \textit{reduced computational complexity (RCC)} alternative to each of the samplers is also designed by amending the algorithms such that, at time $t_n$, the MCMC sweep operates over, at most, 20 changepoints, i.e. over the interval $\left[\phi_{k_{t_n}-19},t_n\right)$.
Due to the well-known path degeneracy problem in SMC (see Kantas et al.~(2011)), the estimates will be poor approximations of the true values, when including static parameters and extending the space of the point process for a long time. We note, at least for our application, it is reasonable to consider $T$ fixed and thus, this is less problematic.

\section{The Finance Problem Revisited}\label{sec:financefinal}

We now return to the example from Section \ref{sec:motivate}
and the settings as in Section \ref{sec:simosfinance1}.

\subsection{Simulated Data}

The saturated and tempered samplers, as well as their RCC alternatives, were implemented using the simulated data set (in Section \ref{sec:simosfinance1}), in order to compare their respective performances against the benchmark sampler and to compare the accuracy of the resulting intensity estimates against an observed intensity process. All of the alternative samplers were implemented under the same conditions, using the algorithm and model parameters as described for the implementation of the benchmark sampler. All results are averaged over 10 runs of the algorithm.
 
In assessing the performance of the sampler, quantities of interest are, once again, the resampling rate and the processing time, as well as the minimum ESS recorded throughout the execution of the sampler.  The resampling rates for all three samplers and their RCC alternatives are presented in Table \ref{Table_ResamplingRateComparisons}, with the corresponding minimum ESS's attained recorded in Table \ref{Table_MinimumESSComparisons} and the corresponding processing times in Table \ref{Table_ProcessingTimeComparisons}. 
Figure \ref{ESS2} displays the evolution of the ESS over a particular run of the algorithm. Figure \ref{ESS3} shows the estimated intensity at each time $t_n$, given data up to time $t_n$.
From Table \ref{Table_ResamplingRateComparisons}, it is clear to see that, for the saturated and tempered samplers, an increase in $M$ results in a decrease in the resampling rates, i.e. a decrease in sampler degeneracy, as expected. It is also plain to see from Table \ref{Table_MinimumESSComparisons} that, as $N$ increases, so does the minimum ESS, and thus the reliability of the estimates. From Tables \ref{Table_ResamplingRateComparisons} and \ref{Table_MinimumESSComparisons}, Figure \ref{ESS3}
and comparing Figure \ref{ESS2} to Figure \ref{ESS1} it is clear that the saturated and tempered samplers significantly outperformed the benchmark sampler.

\begin{table}
\begin{center}
\begin{tabular}{|r|r|r|r|r|r|r|}
\cline{2-7}
\multicolumn{1}{c}{} & \multicolumn{2}{|c|}{M = 1} & \multicolumn{2}{|c|}{M = 5} & \multicolumn{2}{|c|}{M = 20}\\
\cline{2-7}
\multicolumn{1}{c}{} & \multicolumn{1}{|c|}{N=100} & \multicolumn{1}{|c|}{N=1000} & \multicolumn{1}{|c|}{N=100} & \multicolumn{1}{|c|}{N=1000} & \multicolumn{1}{|c|}{N=100} & \multicolumn{1}{|c|}{N=1000}\\
\hline
Benchmark & 31.3\% & 52.0\% & 42.3\% & 94.4\% & 74.0\% & 99.7\%\\
Benchmark - RCC & 37.6\% & 88.1\% & 69.0\% & 99.7\% & 99.4\% & 99.7\%\\
Saturated & 21.0\% & 21.3\% & 19.7\% & 20.1\% & 18.2\% & 17.6\%\\
Saturated - RCC & 20.7\% & 20.7\% & 18.5\% & 18.8\% & 15.4\% & 15.4\%\\
Tempered & 2.0\% & 2.0\% & 1.9\% & 1.9\% & 1.7\% & 1.7\%\\
Tempered - RCC & 2.0\% & 2.0\% & 1.7\% & 1.8\% & 1.4\% & 1.4\%\\
\hline
\end{tabular}
\caption{Table showing the resampling rates of each of the three SMC samplers  and their reduced computational complexity alternatives, for the six algorithm parameterisations that were tested. The ESS plots for the saturated and tempered samplers with $N=1000$, $M=5$ are given in Figure \ref{ESS2} for comparison with the corresponding ESS plot for the benchmark sampler given in Figure \ref{ESS1}}
\label{Table_ResamplingRateComparisons}
\end{center}
\end{table}

\begin{table}[tbp]
\begin{center}
\begin{tabular}{|r|r|r|r|r|r|r|}
\cline{2-7}
\multicolumn{1}{c}{} & \multicolumn{2}{|c|}{M = 1} & \multicolumn{2}{|c|}{M = 5} & \multicolumn{2}{|c|}{M = 20}\\
\cline{2-7}
\multicolumn{1}{c}{} & \multicolumn{1}{|c|}{N=100} & \multicolumn{1}{|c|}{N=1000} & \multicolumn{1}{|c|}{N=100} & \multicolumn{1}{|c|}{N=1000} & \multicolumn{1}{|c|}{N=100} & \multicolumn{1}{|c|}{N=1000}\\
\hline
Benchmark & 1.0 & 1.0 & 1.0 & 1.0 & 1.0 & 1.0\\
Benchmark - RCC & 1.0 & 1.0 & 1.0 & 1.0 & 1.0 & 1.0\\
Saturated & 38.1 & 410.2 & 38.6 & 397.0 & 38.6 & 398.9\\
Saturated - RCC & 38.5 & 401.2 & 40.6 & 394.4 & 43.0 & 425.9\\
Tempered & 47.6 & 484.7 & 47.7 & 475.5 & 47.9 & 483.4\\
Tempered - RCC & 47.8 & 475.7 & 48.4 & 481.7 & 48.3 & 486.6\\
\hline
\end{tabular}
\caption{Table showing the minimum ESS encountered during implementation by each of the three SMC samplers  and their reduced computational complexity alternatives, for the six algorithm parameterisations that were tested.}
\label{Table_MinimumESSComparisons}
\end{center}
\end{table}

\begin{table}
\begin{center}
\begin{tabular}{|r|r|r|r|r|r|r|}
\cline{2-7}
\multicolumn{1}{c}{} & \multicolumn{2}{|c|}{M = 1} & \multicolumn{2}{|c|}{M = 5} & \multicolumn{2}{|c|}{M = 20}\\
\cline{2-7}
\multicolumn{1}{c}{} & \multicolumn{1}{|c|}{N=100} & \multicolumn{1}{|c|}{N=1000} & \multicolumn{1}{|c|}{N=100} & \multicolumn{1}{|c|}{N=1000} & \multicolumn{1}{|c|}{N=100} & \multicolumn{1}{|c|}{N=1000}\\
\hline
Benchmark & 612.9 & 9689.1 & 2849.7 & 45690.4 & 13352.1 & 144621.3\\
Benchmark - RCC & 449.0 & 7910.9 & 1132.7 & 10657.6 & 3106.2 & 31208.5\\
Saturated & 1125.3 & 10667.8 & 3234.3 & 39061.1 & 15381.9 & 141817.3\\
Saturated - RCC & 637.5 & 6215.2 & 1200.7 & 11412.6 & 4391.9 & 47662.8\\
Tempered & 1160.2 & 10633.4 & 3138.4 & 38679.6 & 14086.7 & 130899.1\\
Tempered - RCC & 666.0 & 6424.4 & 1156.3 & 11209.1 & 3231.3 & 34795.3\\
\hline
\end{tabular}
\caption{Table showing the processing time, in seconds, for each of the three samplers and their reduced computational complexity alternatives, for the six algorithm parameterisations that were tested.}
\label{Table_ProcessingTimeComparisons}
\end{center}
\end{table}

\begin{figure}[h!]
\begin{center}
\subfloat[Saturated]{\label{ESS2:Saturated}\includegraphics[width=0.5\textwidth]{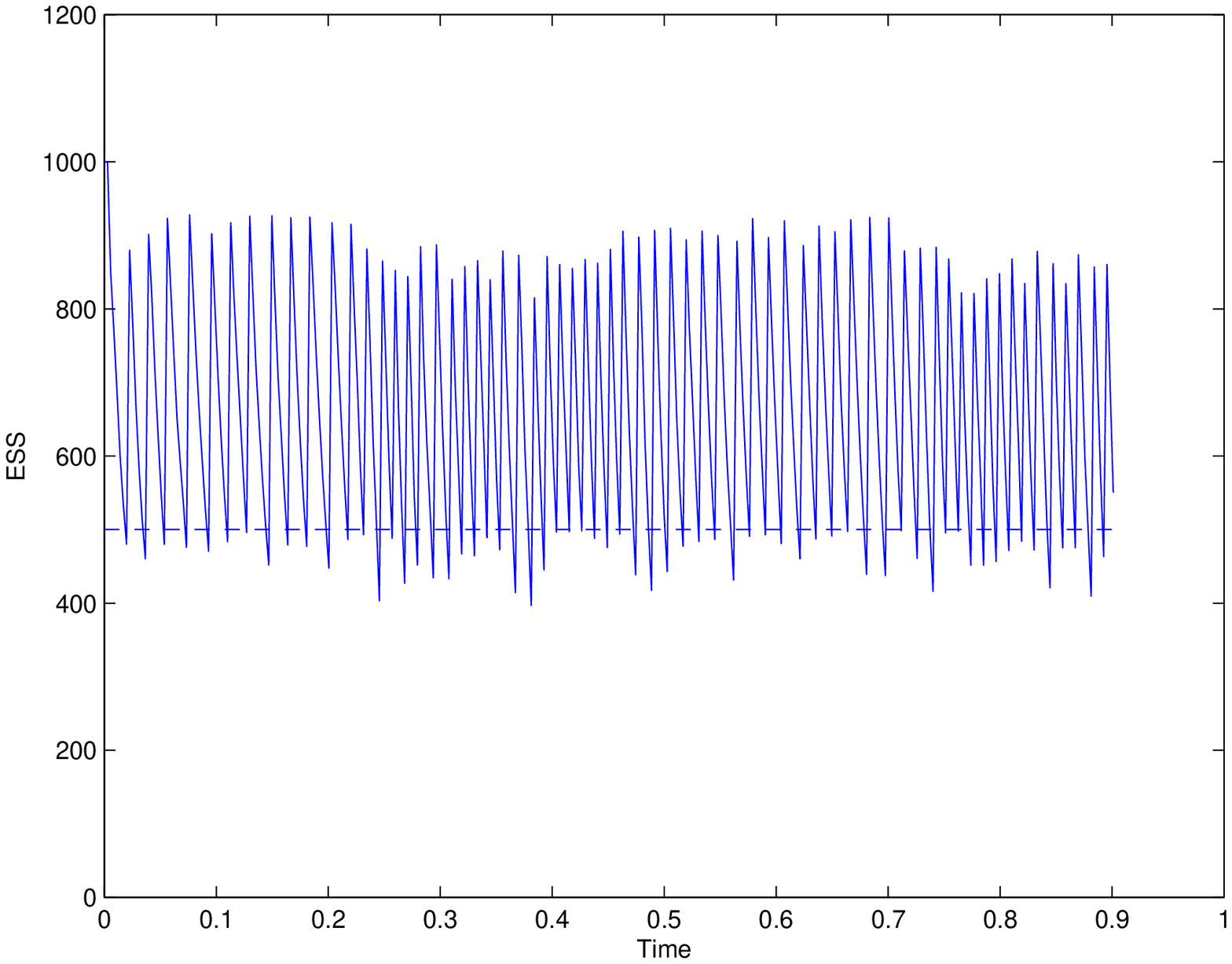}}
\subfloat[Tempered]{\label{ESS2:Tempered}\includegraphics[width=0.5\textwidth]{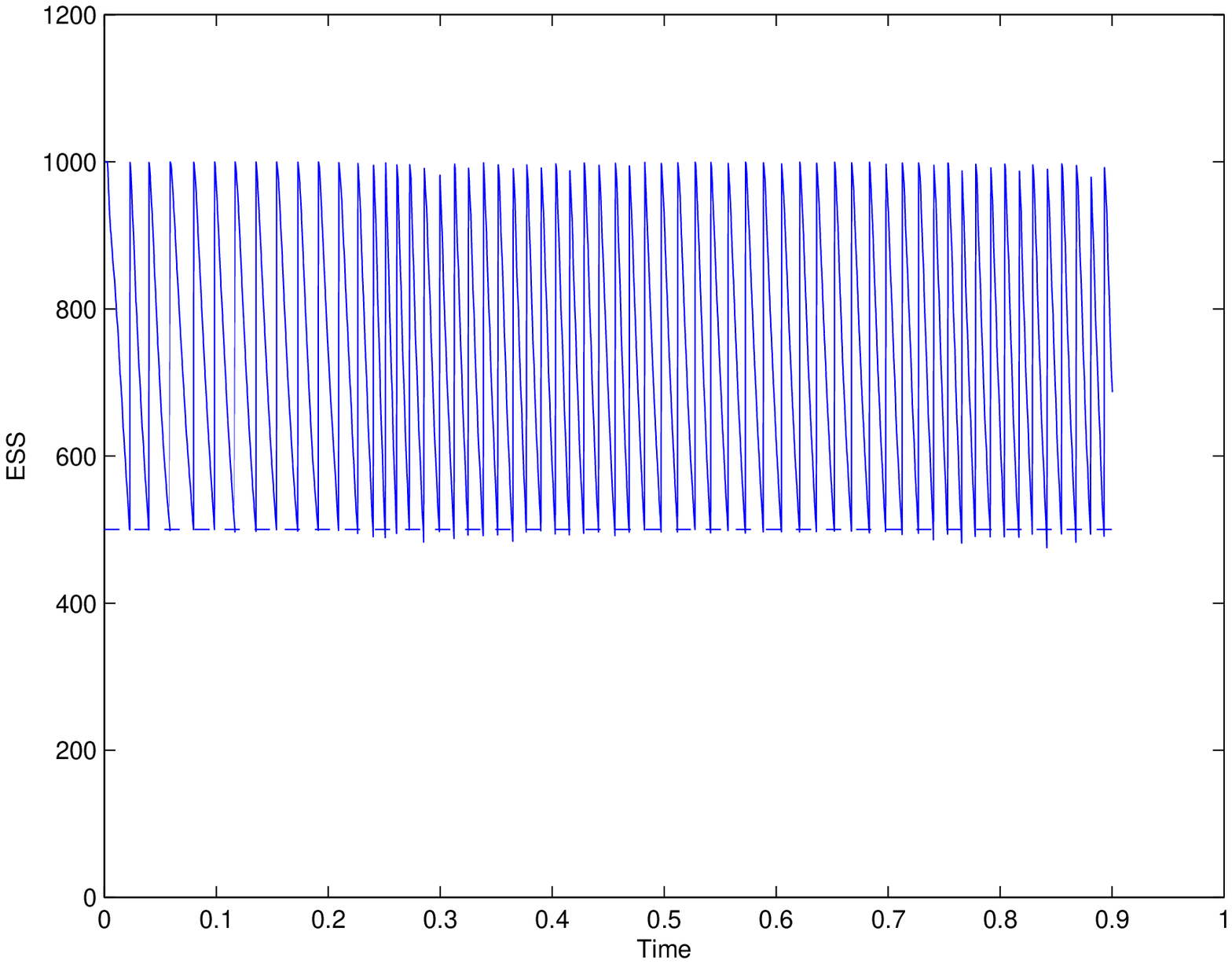}}
\caption{Effective Sample Size plots for the SMC samplers with state space saturation (left) and data point tempering (right), run with $N=1000$ particles and with $M=5$ MCMC sweeps at each iteration. The dashed line indicates the resampling threshold at $N/2=500$ particles; the corresponding resampling rates are 20.1\% for the saturated sampler and 1.9\% for the tempered sampler.}
\label{ESS2}
\end{center}
\end{figure}

We use the posterior medians to report intensities.
Since we have access to a `true' intensity process, the accuracy of these estimated intensity process is measured using the root mean square error (RMSE). Table \ref{Table_FilteredRMSEComparisons} presents the RMSEs of the  intensity estimates (given the data up to $t_n$, averaged over each $t_n$) and Table \ref{Table_SmoothedRMSEComparisons} presents the RMSEs of the smoothed (conditional upon the entire data set) intensity estimates resulting from each of the three samplers and their RCC alternatives. 
The most important result to note is the performance of the saturated and tempered samplers in comparison with the unaltered sampler. As can be seen in terms of accuracy for intensity estimates, the two proposed alterations to the sampler improve the performance consistently and significantly. Looking at the resampling rates and processing times, in Tables \ref{Table_ResamplingRateComparisons} and \ref{Table_ProcessingTimeComparisons} respectively, we can see that, as expected, although the tempered sampler resampled the particles significantly less than the benchmark sampler, the individual incorporation of each data point resulted in a greater computational cost. These two aspects of the benchmark and tempered samplers appear to have countered each other, resulting in their processing times being largely similar.

We consider also the effect that changes in $M$ and $N$ have on the accuracy of estimates provided by the saturated and tempered samplers. For the saturated and tempered samplers, the results in Tables \ref{Table_FilteredRMSEComparisons} and \ref{Table_SmoothedRMSEComparisons} corroborate the expected improvement in accuracy, in both for the sequential estimates at $t_n$ given data up-to $t_n$
and smoothed estimates (given the entire data), that results from an increase in the number of particles used. Whilst for the sequential estimates, there is no clear improvement in accuracy with increasing $M$, an improvement can be seen in the accuracy of the smoothed estimates.

\begin{table}[h!]
\begin{center}
\begin{tabular}{|r|r|r|r|r|r|r|}
\cline{2-7}
\multicolumn{1}{c}{} & \multicolumn{2}{|c|}{M = 1} & \multicolumn{2}{|c|}{M = 5} & \multicolumn{2}{|c|}{M = 20}\\
\cline{2-7}
\multicolumn{1}{c}{} & \multicolumn{1}{|c|}{N=100} & \multicolumn{1}{|c|}{N=1000} & \multicolumn{1}{|c|}{N=100} & \multicolumn{1}{|c|}{N=1000} & \multicolumn{1}{|c|}{N=100} & \multicolumn{1}{|c|}{N=1000}\\
\hline
Benchmark & 688.561 & 1116.639 & 620.432 & 1942.992 & 1330.232 & 1501.263\\
Benchmark - RCC & 676.932 & 2026.956 & 880.824 & 2247.313 & 1472.126 & 1264.533\\
Saturated & 242.834 & 192.580 & 228.390 & 193.778 & 237.315 & 198.223\\
Saturated - RCC & 229.449 & 189.279 & 224.692 & 193.379 & 225.592 & 194.623\\
Tempered & 254.396 & 196.928 & 247.754 & 201.681 & 248.367 & 202.501\\
Tempered - RCC & 256.012 & 191.407 & 227.241 & 197.043 & 230.805 & 200.227\\
\hline
\end{tabular}
\caption{Table showing the root mean square error of the intensity.
This is given the data up to $t_n$, averaged over each $t_n$ and
for each of the three samplers and their reduced computational complexity alternatives, for the six algorithm parameterisations that were tested.}
\label{Table_FilteredRMSEComparisons}
\end{center}
\end{table}

\begin{table}[h!]
\begin{center}
\begin{tabular}{|r|r|r|r|r|r|r|}
\cline{2-7}
\multicolumn{1}{c}{} & \multicolumn{2}{|c|}{M = 1} & \multicolumn{2}{|c|}{M = 5} & \multicolumn{2}{|c|}{M = 20}\\
\cline{2-7}
\multicolumn{1}{c}{} & \multicolumn{1}{|c|}{N=100} & \multicolumn{1}{|c|}{N=1000} & \multicolumn{1}{|c|}{N=100} & \multicolumn{1}{|c|}{N=1000} & \multicolumn{1}{|c|}{N=100} & \multicolumn{1}{|c|}{N=1000}\\
\hline
Benchmark & 768.702 & 670.656 & 495.019 & 627.909 & 489.243 & 571.107\\
Benchmark - RCC & 698.640 & 1034.890 & 572.794 & 572.841 & 535.004 & 599.031\\
Saturated & 360.794 & 264.331 & 296.953 & 114.064 & 153.444 & 89.397\\
Saturated - RCC & 478.871 & 265.477 & 405.767 & 266.980 & 468.853 & 205.243\\
Tempered & 350.015 & 170.321 & 271.712 & 128.078 & 157.709 & 81.666\\
Tempered - RCC & 485.825 & 249.529 & 475.348 & 193.898 & 514.107 & 180.914\\
\hline
\end{tabular}
\caption{Table showing the smoothed root mean square error 
of the intensity. This is given the entire data set and
for each of the three samplers and their reduced computational complexity alternatives, for the six algorithm parameterisations that were tested.}
\label{Table_SmoothedRMSEComparisons}
\end{center}
\end{table}

\begin{figure}[h!]
\begin{center}
\subfloat[Benchmark]{\label{FilteredEstimates:Benchmark}\includegraphics[width=0.3\textwidth]{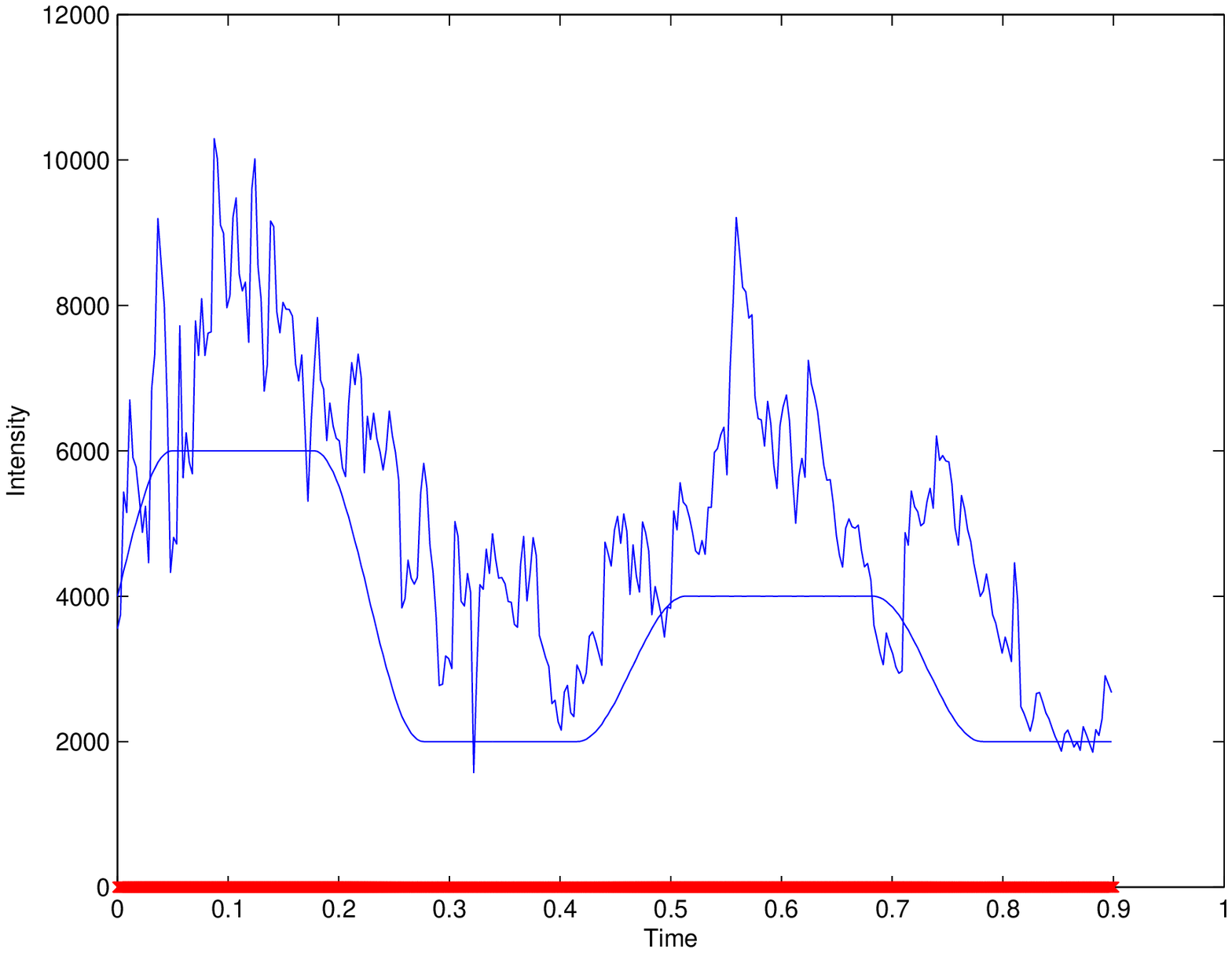}}
\subfloat[Saturated]{\label{FilteredEstimates:Saturated}\includegraphics[width=0.3\textwidth]{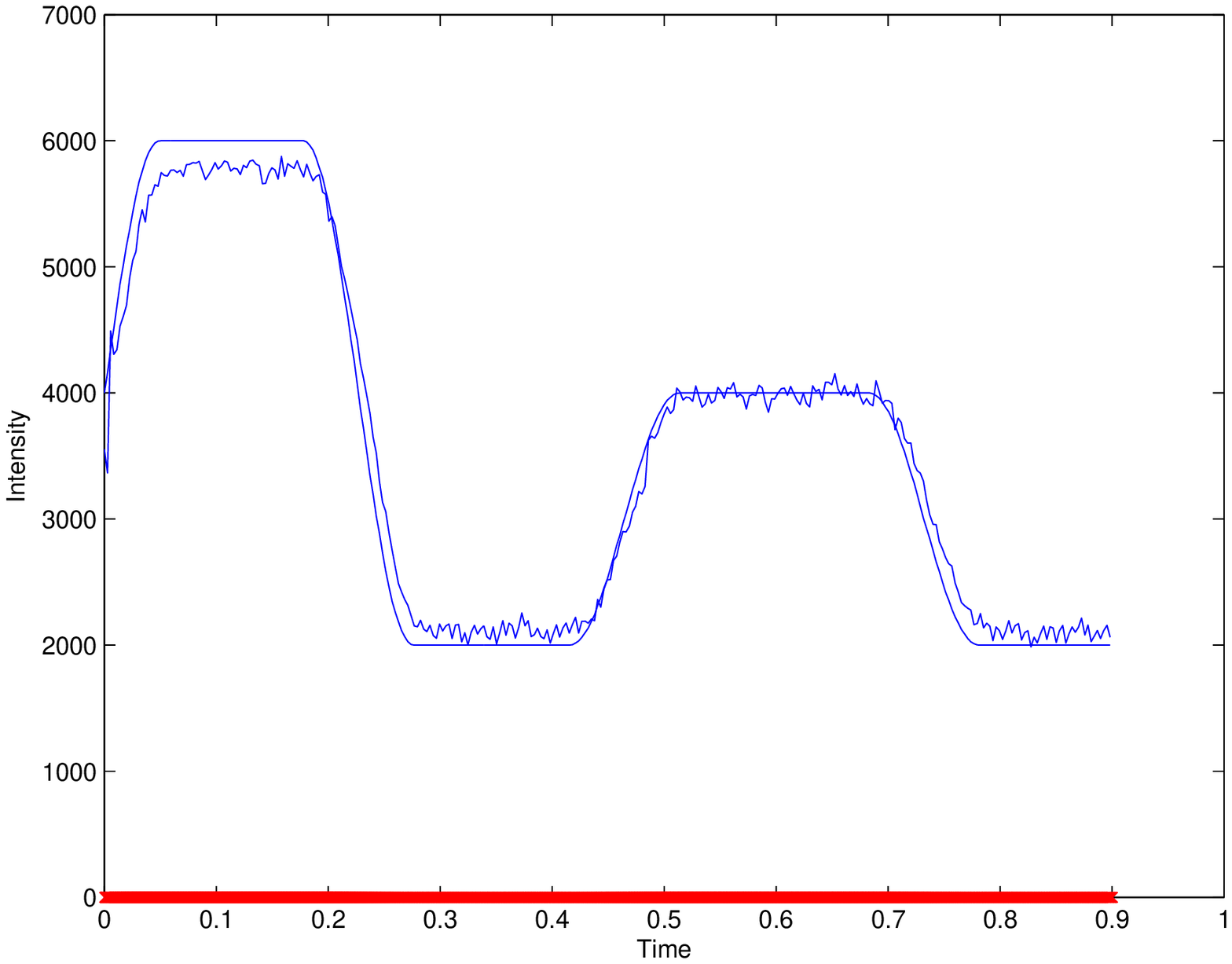}}
\subfloat[Tempered]{\label{FilteredEstimates:Tempered}\includegraphics[width=0.3\textwidth]{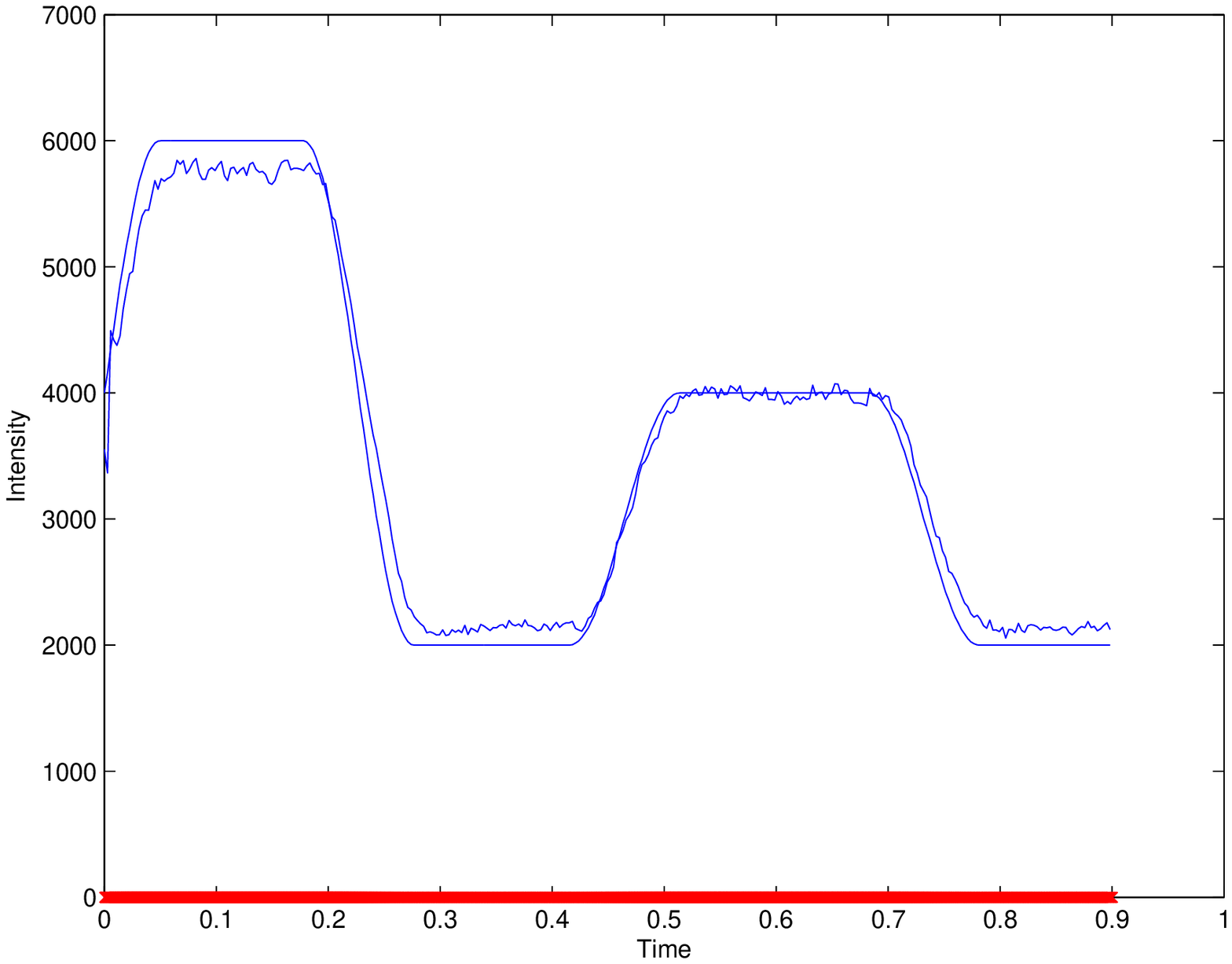}}
\caption{Estimates (given the data up to $t_n$) of the intensity of a simulated data set, generated by the benchmark SMC sampler (left) and the samplers with state space saturation (centre) and data point tempering (right), run with $N=1000$ particles and with $M=5$ MCMC sweeps at each iteration. The model parameters were $\gamma=0.001$, $\nu=150$ and $s=20$.}
\label{ESS3}
\end{center}
\end{figure}

Finally, using the simulated data, we consider the performance of the samplers when limiting the space over which the invariant MCMC kernels are applied, i.e.~the RCC alternatives. As can be seen from Table \ref{Table_FilteredRMSEComparisons}, the RCC alteration does not sacrifice any accuracy in the estimates of the intensity (given the data up to each time $t_n$), however it can be seen from Table \ref{Table_SmoothedRMSEComparisons} that the accuracy of the smoothed intensity estimates is rather poor. This is to be expected, due to path degeneracy; we note that one cannot estimate static parameters with the RCC approach unless the time window $T$ is quite small.

\subsection{Real Data}

All three samplers were also tested on real financial data, with the RCC alternatives also being used to generate intensity estimates, given the data up to $t_n$: the share price of ARM Holdings, plc., traded on the LSE was used. The entire data set was of size $r_T=1819$, $[0,T]=[0,0.3]$ (represents 3/10 of a trading day, that is, 3/10 of 24 hours; the first trade is just after 9am and the last around 16:15.) with $t_n= n*0.001$. 
Genuine financial data is likely to correspond to a more volatile latent intensity process than that which was used to generate the synthetic data set, and so the parameterisation of the target posterior should be chosen such that large jumps in the intensity process are possible, and such that the intensity may also revert quickly to a lower intensity level. Hence, we specify: $\left\{\gamma,\nu,s\right\}=\left\{0.001,500,250\right\}$.
Each of the samplers were run using $N=1000$ particles, applying $M=5$ MCMC sweeps at each iteration, whilst the resampling rates and the minimum ESS obtained for each procedure were monitored to ensure that the algorithms did not collapse.

Clearly, there is no `known' intensity process against which to compare the point-wise estimates produced by the samplers. In addition, any inverse-duration based representation of the intensity against which useful comparisons could be drawn would involve making assumptions on the smoothness of the intensity process itself. Thus, we turn to measuring the one-step-ahead predictive accuracy of the estimators of the intensity. This is achieved as follows: denoting the intensity estimated over the interval $\left[t_{n-j},t_n\right)$ as $\hat\lambda_{n,j}$, one predicts the expected number of ticks in the interval $\left[t_{n+i-j},t_{n+i}\right)$ as $\left(\hat\lambda_{n,j}\right)^{-1}$ for $i\ge1$ and $j\ge1$, where $j$ is the number of periods over which the prediction is made and $i$ is a lag index. The prediction errors are then calculated based on the predicted and observed number of ticks in the period $\left[t_{n+i-j},t_{n+i}\right)$; the root mean square prediction error (RMSPE) will be used.
 We will report on the one-step-ahead estimates ($i=1$), estimating the intensity over each interval with $j=1$.

\begin{table}
\begin{center}
\begin{tabular}{|r|c|c|c|c|}
\cline{2-5}
\multicolumn{1}{c|}{} & Data $t_n$ & Smoothed & Processing & Resampling\\
\multicolumn{1}{c|}{} & RMSPEs & RMSPEs & Times (s) & Rates\\
\hline
Saturated & 2.18876 & 2.13479 & 4064.5 & 39.5\%\\
Saturated - RCC & 2.19112 & - & 2193.1 & 39.9\%\\
Tempered & 2.34671 & 2.11468 & 4605.5 & 19.8\%\\
Tempered - RCC & 2.42776 & - & 2237.3 & 19.9\%\\
\hline
\end{tabular}
\caption{Table showing the root mean square prediction errors for the  intensity estimates (given data up to time $t_n$ and entire data (smoothed))
given by each of the three samplers for the parameter values $N=1000$, $M=5$. The RMSPEs for the smoothed intensity estimates given by the RCC alternatives to the samplers are also provided, along with the observed processing times and resampling rates for each sampler.}
\label{Table_RMSPEComparisons}
\end{center}
\end{table}

Table \ref{Table_RMSPEComparisons} presents the RMSPEs for the intensity estimates resulting from the samplers and the RCC alternatives. It was observed that, in calculating the RMSPEs for lag indices $i=1,\ldots,100$ using each sampler, both the saturated and tempered samplers displayed the smallest error at $i=1$, i.e.~their respective one-step-ahead predictions were more accurate than those made for lags up to 2.64 hours (each observation interval corresponds to 0.0264 days = 1.584 minutes). 

The RCC samplers provide significant computational savings and do not seem to degrade substantially, w.r.t.~the error criteria. Again, we remark that in general one should not trust the estimates of the RCC, but as seen here, they can provide a guideline for the intensity values.
%

\section{Summary}\label{sec:summ}
In this paper we have considered SMC simulation for partially
observed point processes and implemented them for a particular doubly stochastic PP. Two solutions were given,
one based upon saturating the state-space, which is suitable in a wide variety
of applications and data-point tempering which can be used in sequential
problems. We also discussed RCC versions of these algorithms, which reduce computation,
but will be subject to the path degeneracy problem when including static parameters and considering the smoothing distribution.
We saw that the methods can be successful, in terms of
weight degeneracy versus the benchmark approach detailed in Del Moral et al.~(2007). In addition, for real data it was observed that predictions using the RCC
could be reasonable (relative to the normal versions of the algorithms), but caution on using these estimates should be used.

The methodology we have presented is not online. As we have seen, when one modifies the approaches to have fixed computational complexity, the path degeneracy problem occurs and one cannot deal with scenario with static parameters.
In this case, we are working with Dr. N. Whiteley on a technique based upon fixed window filtering.  This is an on-line
algorithm which allows data to be incorporated as they arrive with
computational cost which is non-increasing over time, but is biased. The approach involves sampling from a sequence of
distributions which are constructed such that, at time $t_n$,
previously sampled events in $[0,t_{n-\ell}]$ can be discarded.
In order to be exact (in the sense of targeting the
true posterior distributions), this scheme would involve
point-wise evaluation of an intractable density. We are working on
a sensible approximation of this density, at the cost of
introducing a small bias.

\subsubsection*{Acknowledgement}
We thank Nick Whiteley for conversations on this work. The second author was supported by an MOE grant.
We thank two referees for their comments, which have vastly improved the article.

\section*{Appendix}

\subsection*{Proposition \ref{prop:cmjustification}}

In this appendix we give a proof of
Proposition \ref{prop:cmjustification}.
For probability meaure $\varpi$ and function $f$, $\varpi(f) := \int f(x)\varpi(dx)$.
For any collection of points $(\chi_1^{(1)},\dots,\chi_{n-1}^{(N)})\in
\bar{E}_{n-1}^N$ write
$$
S_{n-1}^N(x) = \frac{1}{N}\sum_{i=1}^N\mathbb{I}_{\{\chi_{n-1}^{(i)}\}}(x).
$$
The transition kernels are written $K_1$ (which is not to be confused with
the $K_1$ from the SMC samplers algorithm) and
for any $n\geq 2$, $N\geq 1$, $N-$empirical density $S_{n-1}^N$,
$K_{S_{n-1}^N,n}$ is the kernel of invariant distribution
$$
\frac{l_{(t_{n-1},t_{n}]}(\bar{y}_{n,1};\bar{x}_n)
}{p_n(\bar{y}_{n,1}|\bar{y}_{n-1})}
\mathsf{p}(\bar{x}_{n,1})
S_{n-1}^N(\bar{x}_{n-1}).
$$
Recall the generic notation $\bar{x}_n\in \bar{E}_n$.
We drop the dependence upon the data and denote
\begin{equation}
g_n(\bar{x}_n) = \frac{l_{(t_{n-1},t_{n}]}(\bar{y}_{n,1};\bar{x}_n)
}{p_n(\bar{y}_{n,1}|\bar{y}_{n-1})}
\mathsf{p}(\bar{x}_{n,1})
\label{eq:filteringweight}.
\end{equation}
The $N-$empirical measure of points generated up to time $n-1$ is written
$S_{\bar{x},n-1}^N$.
For a given $n\geq 1$, $f_n:\bar{E}_n\rightarrow\mathbb{R}$
we have the notation $K_{n,S_{n-1}^N}(f_n)(x):= \int_{\bar{E}_n}f_n(y)K_{n,S_{n-1}^N}(x,dy)$
and $i\geq 1$, $K^i_{n,S_{n-1}^N}(f_n)(x) := \int_{\bar{E}_n} K^{i-1}_{n,S_{n-1}^N}(x,dy)K_{n,S_{n-1}^N}(f_n)(y)$,
$K^0_{n,S_{n-1}^N}(x,dy)=\delta_x(dy)$ the Dirac measure. The $\sigma-$finite
measure $d\bar{x}_{n,1}$ is defined on the space $\bar{E}_n\setminus\bar{E}_{n-1}$;
in practice it is the product of an appropriate version of Lebesgue
and counting measures.

The following
assumption is made.
\begin{quote}
\emph{
Assumption \emph{\bf{(A)}}. There exist an $\epsilon_1\in(0,1)$ and probability measure
$\kappa_1$ on $\bar{E}_1$ such that for any $\bar{x}_1\in \bar{E}_1$
$$
K_1(\bar{x}_1,\cdot) \geq \epsilon_1\kappa_1(\cdot).
$$
For any $n\geq 2$, there exist an $\epsilon_n\in(0,1)$ and probability measure
$\kappa_n$ on $\bar{E}_n\setminus \bar{E}_{n-1}$
such that for any $\widetilde{x}_n\in \bar{E}_n$
and any collection of points $(\chi_{n-1}^{(1)},\dots,\chi_{n-1}^{(N)})\in
\bar{E}_{n-1}^N$
$$
K_{S_{n-1}^N,n}(\bar{x}_n,\cdot) \geq \epsilon_n S_{n-1}^N(\cdot)\kappa_n(\cdot).
$$
For any $n\geq 2$
$$
\sup_{\bar{x}_{n-1}\in \bar{E}_{n-1}}\int_{\bar{E}_n\setminus\bar{E}_{n-1}} |g_n(\bar{x}_{n-1},\bar{x}_{n,1})| d\bar{x}_{n,1} <+\infty
$$
where $g_n$ is as in (\ref{eq:filteringweight}).
}
\end{quote}

It should be noted that the uniform ergodicity assumption on $K_{S_{n-1}^N,n}(\bar{x}_n,\cdot)$
is quite strong. If the kernel $K_{S_{n-1}^N,n}$ were an Metropolis-Hastings independence sampler of proposal $S_{n-1}^N\times q_n(\cdot)$ $\bar{x}_n=(\bar{x}_{n-1},\bar{x}_{n,1})$, then
$$
K_{S_{n-1}^N,n}(\bar{x}_n,\cdot) \geq \min\bigg\{1,
\frac{g_n(\bar{v}_n)q_n(\bar{x}_{n,1})}{g_n(\bar{x}_n)q_n(\bar{v}_{n,1})}\bigg\}S_{n-1}^N(\cdot)q_n(\cdot)
$$
satisfies the assumption if $q_n(\bar{x}_{n,1})/g_n(\bar{x}_n)$ is uniformly
lower-bounded. Note also, due to the suppression of the data from the notation,
it is typical that $\epsilon_n$ would depend upon $\bar{y}_{n}$.

\begin{proof}
The proof is inductive on $n$. Some details are omitted as the
proof is quite similar to the control of adaptive MCMC chains, e.g.~Andrieu et al.~(2011). It should be noted the proof for this algorithm differs
as the kernel possesses an invariant measure that does not change
with the iteration $i\in\{1,\dots,N\}$.

Let $n=1$ then, by (A) $K_1$, is a uniformly
ergodic Markov kernel of invariant measure $\pi_1$. It is simple to use the
Poisson equation to prove the proposition, which is given to establish the induction.
Let $\hat{f}_1(\bar{x}_1)=\sum_{i=0}^{\infty}[K_1^i(f_1)(\bar{x}_1)-\pi_1(f_1)]$
be the solution to the Poisson equation; $\hat{f}_1 - K_1(\hat{f}_1) = f_1 - \pi_1(f_1)$. Then
\begin{eqnarray*}
\sum_{i=1}^N[f_1(\bar{x}_1^{(i)})-\pi_1(f_1)] & = &
\sum_{i=1}^N[\hat{f}_1(\bar{x}_1^{(i)})-K_1(\hat{f}_1)(\bar{x}_1^{(i)})]\\
& = & \sum_{i=1}^{N-1}[\hat{f}_1(\bar{x}_1^{(i+1)})-K_1(\hat{f}_1)(\bar{x}_1^{(i)})]
+ \hat{f}_1(\bar{x}_1^{(1)}) -
K_1(\hat{f}_1)(\bar{x}_1^{(N)})
\end{eqnarray*}
the first quantity on the R.H.S. is a Martingale, $M_N^1$, w.r.t.~the filtration $\mathcal{F}_1^{i}$ (i.e.~the $\sigma-$algebra
generated by Markov chain). Then, using the Minkowski inequality
$$
\mathbb{E}_{\bar{x}_1^{(1)}}\bigg[\bigg|\frac{1}{N}\sum_{i=1}^N[f_1(\bar{x}_1^{(i)})-\pi_1(f_1)]\bigg|^p\bigg]^{1/p}
\leq \frac{1}{N} \bigg\{ \mathbb{E}_{\bar{x}_1^{(1)}}\bigg[\big|M_N^1\big|^p\bigg]^{1/p}
+ |\hat{f}_1(\bar{x}_1^{(1)})| + \mathbb{E}_{\bar{x}_1^{(1)}}\bigg[
\bigg|K_1(\hat{f}_1)(\bar{x}_1^{(N)})\big|^p\bigg]^{1/p}
\bigg\}.
$$
The last term can be dealt with as follows.
\begin{eqnarray*}
\mathbb{E}_{\bar{x}_1^{(1)}}\bigg[
\bigg|K_1(\hat{f}_1)(\bar{x}_1^{(N)})\bigg|^p\bigg]^{1/p} & \leq & \mathbb{E}_{\bar{x}_1^{(1)}}\bigg[\bigg|\sum_{i=0}^{\infty}[K_1^i(f_1)(\bar{x}_1^{(N+1)})
-\pi_1(f_1)]\bigg|^{p}\bigg]^{1/p}\\
& \leq & \|f_1\|
\sum_{i=0}^{\infty}\mathbb{E}_{\bar{x}_1^{(1)}}\bigg[\bigg|[K_1^i-\pi_1]\bigg(\frac{f_1}{\|f_1\|}\bigg)(\bar{x}_1^{(N+1)})\bigg|^{p}\bigg]^{1/p}\\
& \leq & \frac{\|f_1\|}{\epsilon_1}
\end{eqnarray*}
here we have applied the conditional Jensen inequality and the bound on the
total variation distance for uniformly ergodic Markov chains:
$\forall x\in\bar{E}_1$, $\sup_{f:\bar{E}_{1}
\rightarrow[0,1]}|K_1^i(f)(x)-\pi_1(f)|\leq (1-\epsilon_1)^i$.
Note that this bound holds for any $\bar{x}_1\in \bar{E}_1$.
The Martingale term is bounded using the B\"urkholder
and Davis inequalities (i.e.~the inequality below holds for any $p\geq 1$): $$
\mathbb{E}_{\bar{x}_1^{(1)}}\bigg[\big|M_N^1\big|^p\bigg]^{1/p}
\leq B_p \mathbb{E}_{\bar{x}_1^{(1)}}\bigg[\bigg|\sum_{i=1}^{N-1}[\hat{f}_1(\bar{x}_1^{(i)})-
K_1(\hat{f}_1)(\bar{x}_1^{(i)})]^2\bigg|^{p/2}\bigg]^{1/p}.
$$
When $p\geq 2$ the Minkowski inequality and the above
manipulations yield a bound $\sqrt{N}B(p,\epsilon_1)\|f_1\|$, with
$B(p,\epsilon_1)$ a constant only depending upon $p$ and
$\epsilon_1$. When $p\in[1,2)$ the inequality $(a-b)^2 \leq
2(a^2+b^2)$ for $a,b\in\mathbb{R}$ is applied then Jensen to yield
a similar bound; see Andrieu et al.~(2011) and the references therein. Thus,
for $n=1$ it follows
$
\mathbb{E}_{\bar{x}_1^{(1)}}\big[|M_N^1|^p\big]^{1/p}
\leq
\sqrt{N}B(p,\epsilon_1)\|f_1\|
$; note that $B(p,\epsilon_1)$ depends only on $\epsilon_1$ and $p$ - this is important
in the sequel. Putting these bounds together and noting that, by the above
arguments, the solution to the Poisson equation is uniformly bounded in $x$
the proof at rank $n=1$ is completed.

Now assume the result at $n-1$ and consider $n$. Note that via Fubini
\begin{equation*}
\pi_n(f_n) = \int_{\bar{E}_{n}} f_n(\bar{x}_{n}) g_n(\bar{x}_{n})\pi_{n-1}(d\bar{x}_{n-1})d\bar{x}_{n,1}\\
 = \int_{\bar{E}_{n-1}} I(f_n\times g_n)(\bar{x}_{n-1}) \pi_{n-1}(d\bar{x}_{n-1})
\end{equation*}
where $I(f_n\times g_n) = \int_{\bar{E}_n\setminus E_{n-1}} f_n(\bar{x}_{n-1},\bar{x}_{n,1})
g_n(\bar{x}_{n-1},\bar{x}_{n,1})
d\bar{x}_{n,1}$. Then application of the Minkowski inequality yields:
\begin{eqnarray}
\mathbb{E}_{\bar{x}_{1}^{(1)}}\bigg[\bigg|\frac{1}{N}\sum_{i=1}^Nf_n(\bar{x}_{n}^{(i)})-\pi_n(f_n)\bigg|^p\bigg]^{1/p}
& \leq  &
\mathbb{E}_{\bar{x}_{1}^{(1)}}\bigg[\bigg|\frac{1}{N}\sum_{i=1}^Nf_n(\bar{x}_{n}^{(i)})-S_{\bar{x},n-1}^N(I(f_n\times
g_n))\bigg|^p\bigg]^{1/p}
\nonumber\\ & &
+
\mathbb{E}_{\bar{x}_{1}^{(1)}}\bigg[\bigg|[S_{\bar{x},n-1}^N-\pi_{n-1}](I(f_n\times
g_n))\bigg|^p\bigg]^{1/p}
\label{eq:induction}.
\end{eqnarray}
Due to the induction hypothesis and (A), the second term on the R.H.S of the inequality is upper-bounded by
$$
\frac{B_{p,n-1}\sup_{\bar{x}_{n-1}\in \bar{E}_{n-1}}I(|f_n\times g_n|)(\bar{x}_{n-1})}{\sqrt{N}}
\leq
\frac{B_{p,n}\|f_n\|}{\sqrt{N}}
$$
for some $B_{p,n}<+\infty$; if the data were not suppressed, then there is
an explicit dependence upon this quantity.
Then considering the first term on the R.H.S of (\ref{eq:induction}), conditioning
upon the $\sigma-$algebra $\mathcal{F}_1^N\otimes\cdots\otimes\mathcal{F}_{n-1}^{N}$
generated by the process at time $n$ is a uniformly ergodic Markov chain of invariant
distribution $S_{\bar{x},n-1}^N(d\bar{x}_{n-1})g_n(\bar{x}_n)d\bar{x}_{n,1}$. Thus, for example:
$$
\mathbb{E}_{\bar{x}_1^{(1)}}\bigg[\bigg|K_{n,S_{\bar{x},n-1}^N}(\hat{f}_n)(\bar{x}_n^{(N)})\bigg|^p\bigg]^{1/p}
\leq \mathbb{E}_{\bar{x}_1^{(1)}}\bigg[\bigg(\frac{\|f_n\|}{\epsilon_n}\bigg)^p\bigg]^{1/p}
$$
adopting exactly the above arguments. Noting that the bound on the conditional
expectation is deterministic,
i.e.~does not depend upon $\mathcal{F}_1^N\otimes\cdots\otimes\mathcal{F}_{n-1}^{N}$,
the induction is easily completed.
\end{proof}

\subsection*{Proof of Proposition \ref{prop:data_point}}

For the proof of Proposition \ref{prop:data_point}, we require a round of notations. We write
$\tilde{x}_p=(\bar{x}_p,\bar{x}_p')\in\bar{E}_1^2$ and define the following quantities:
$$
G_p(\tilde{x}_p) = \frac{\pi_p(\bar{x}_p)}{\pi_{p-1}(\bar{x}_p)} \quad 1\leq p \leq r_1
$$
with $G_0(\tilde{x}_0)=1$. In addition, set $\eta_0(\cdot)=\mathsf{p}(\cdot)$ and
$$
M_p(\tilde{x}_{p-1},d\tilde{x}_p) = \delta_{x_{p-1}'}(dx_p)K_{p}(x_p,dx_p') \quad 1\leq p \leq r_1
$$
We add an extra Markov kernel to allow us to use directly formulae in Del Moral (2004); 
$M_{r_1+1}(\tilde{x}_{r_1},d\tilde{x}_{r_{1}+1})=\delta_{\tilde{x}_{r_1}}(d\tilde{x}_{r_{1}+1})$. Then we define
$$
\eta_p(d\tilde{x}_p) = \frac{\int_{(\bar{E}_1^2)^{p-1}} \eta_0(d\tilde{x}_0) \prod_{q=0}^{p-1} G_q(\tilde{x}_q)M_q(\tilde{x}_{q-1},d\tilde{x}_q)}{\int_{(\bar{E}_1^2)^{p}} \eta_0(d\tilde{x}_0) \prod_{q=0}^{p-1} G_q(\tilde{x}_q)M_q(\tilde{x}_{q-1},d\tilde{x}_q)}
\quad 1 \leq p \leq r_1+1.
$$
In addition $Q_p(\tilde{x}_{p-1},d\tilde{x}_{p})=G_{p-1}(\tilde{x}_{p-1})M_{p}(\tilde{x}_{p-1},d\tilde{x}_{p})$, $1\leq p \leq r_1+1$, with
$$
Q_{p,n}(\tilde{x}_{p-1},d\tilde{x}_{n}) = \int_{(\bar{E}_1^{2})} Q_{p+1}(\tilde{x}_{p},d\tilde{x}_{p+1})\dots Q_n(\tilde{x}_{n-1},\tilde{x}_n) \quad 1\leq p\leq n \leq r_1+1
$$
with the convention that $Q_{p,p}$ is the identity operator. Also define $P_{p,n}(\tilde{x}_{p-1},d\tilde{x}_{n})=Q_{p,n}(\tilde{x}_{p-1},d\tilde{x}_{n})/Q_{p,n}(1)(\tilde{x}_{p-1})$ and finally
$$
\overline{Q}_{p,n}(\tilde{x}_{p-1},d\tilde{x}_{n}) = \frac{Q_{p,n}(\tilde{x}_{p-1},d\tilde{x}_{n})}{\eta_pQ_{p,n}(1)}.
$$

\begin{proof}[Proof of Proposition \ref{prop:data_point}]
We have from Proposition 9.4.2 of Del Moral (2004) that:
\begin{equation}
\sigma^2_{\textrm{TE},r_1}(f) = \sum_{p=0}^{r_1+1} \eta_p\big(\overline{Q}_{p,r_1+1}(f-\eta_{r_1+1}(f))^2\big).
\label{eq:asymp_var}
\end{equation}
The objective is to re-write the summand in terms of a difference 
$P_{p,r_1+1}(\tilde{x},\cdot)-P_{p,r_1+1}(\tilde{x}',\cdot)$ and use the mixing conditions to control the Dobrushin coefficient of 
the kernel $P_{p,r_1+1}$; see e.g.~Del Moral et al.~(2012) section 4. To that end, we can only consider the first $r_1-1$ terms,
for which the Dobrushin coeffient will satisfy:
\begin{equation}
\beta(P_{p,r_1+1}) := \sup_{\tilde{x},\tilde{x}'}\|P_{p,r_1+1}(\tilde{x},\cdot)-P_{p,r_1+1}(\tilde{x}',\cdot)\|_{tv} \leq  (1-\rho)^{\lfloor [r_1+1-p]/2 \rfloor} \quad r_1-p \geq 1 \label{eq:dobrushin_bd}
\end{equation}
for some $\rho\in(0,1)$ that does not depend upon $r_1$ and $\|\cdot\|_{tv}$ the total variation distance (again see Del Moral et al.~(2012), as the condition $(\mathcal{M})_2$ of that paper is satisfied). The reminder of the terms in the sum are easily bounded, independently of $r_1$, and we omit these calculations.

By using standard properties of Feynman-Kac formula, we have that each summand in \eqref{eq:asymp_var} is equal to
$$
\eta_p\bigg(\frac{Q_{p,r_1+1}(1)^2}{\eta_p(Q_{p,r_1+1}(1))^2}\frac{\eta_p(Q_{p,r_1+1}(1)[P_{p,r_1+1}(f)(\tilde{x})-P_{p,r_1+1}(f)])^2}{\eta_p(Q_{p,r_1+1}(1))^2}\bigg)
$$
By using Jensen's inequality, it follows that
$$
\eta_p\bigg(\frac{Q_{p,r_1+1}(1)^2}{\eta_p(Q_{p,r_1+1}(1))^2}\frac{\eta_p(Q_{p,r_1+1}(1)[P_{p,r_1+1}(f)(\tilde{x})-P_{p,r_1+1}(f)])^2}{\eta_p(Q_{p,r_1+1}(1))^2}\bigg) 
$$
$$
\leq
\|f\|^2\beta(P_{p,r_1+1})^2 \frac{\eta_p(Q_{p,r_1+1}(1))^2)^2}{\eta_p(Q_{p,r_1+1}(1))^4}.
$$
Using the fact that (see e.g.~section 4 of Del Moral et al.~(2012))
$$
\sup_{\tilde{x},\tilde{x}'} \frac{Q_{p,r_1+1}(1))(\tilde{x})}{Q_{p,r_1+1}(1))(\tilde{x}')} \leq B
$$
for a $B\in(0,+\infty)$ that does not depend on $r_1$ and using the bound  in \eqref{eq:dobrushin_bd} we can conclude.
\end{proof}

\normalsize

\vspace{0.025 in}

{\ \nocite{*} \centerline{ REFERENCES}
\begin{list}{}{\setlength{\itemindent}{-0.3in}}

\item
{\sc Andrieu}, C., {\sc Jasra}, A., {\sc Doucet}, A. \& {\sc Del Moral}, P. (2011). On non-linear Markov chain Monte Carlo. {\it Bernoulli}, {\bf 17}, 987-1014.

\item
{\sc Barndorff-Nielsen}, O. \& {\sc Shephard}, N. (2001). Non-Gaussian
Ornstein-Uhlenbeck models and some of their uses in financial
economics
(with discussion). \emph{J. R. Statist. Soc. Ser. B},
\textbf{63}, 167--241.

\item
{\sc Beskos}, A., {\sc Crisan}, D. \& {\sc Jasra}, A.~(2011). On the stability of sequential Monte Carlo methods in high dimensions. Technical Report, Imperial College London.

\item
{\sc Centanni}, S. \& {\sc Minozzo}, M. (2006a). A Monte Carlo
approach to filtering for a class of marked doubly stochastic Poisson
processes. \emph{J. Amer. Statist. Assoc.}, \textbf{101}, 1582--1597.

\item
{\sc Centanni}, S. \& {\sc Minozzo}, M. (2006b).
Estimation and filtering by reversible jump MCMC for a doubly stochastic Poisson model for ultra-high-frequency financial data. \emph{J. Statist.
Mod.}, \textbf{6}, 97--118.

\item
{\sc Chopin}, N. (2002). A
sequential particle filter for static models. \emph{Biometrika}, \textbf{89}, 539--552.

\item
{\sc Chopin}, N., {\sc Jacob}, P.~\& {\sc Papaspiliopoulos} O. (2011). SMC$^2$: A sequential Monte Carlo algorithm with particle Markov chain Monte Carlo updates.
Technical Report, CREST-ENSAE.


\item
{\sc Daley}, D. J. \& {\sc Vere-Jones}, D.~(1988). {\it
Introduction to the Theory of Point Processes},
Springer-Verlag: New York.

\item
{\sc Del Moral,} P., (2004). {\it
Feynman-Kac formulae. Genealogical and interacting particle systems},
Springer-Verlag: New York.

\item
{\sc Del Moral,} P., {\sc Doucet}, A. \& {\sc Jasra}, A.~(2006).
Sequential Monte Carlo samplers, {\it J. R. Statist. Soc. B} {\bf 68}, 411-32.

\item
{\sc Del Moral,} P., {\sc Doucet}, A. \& {\sc Jasra}, A.~(2007).
Sequential Monte Carlo for Bayesian computation (with discussion).
\emph{Bayesian Statistics 8}, Ed.~Bayarri, S., Berger, J. O., Bernardo, J. M., Dawid, A. P., Heckerman, D. Smith, A. F. M. and West, M. 115-149,
OUP: Oxford.

\item
{\sc Del Moral,} P., {\sc Doucet}, A. \& {\sc Jasra}, A.~(2012).
On adaptive resampling procedures for sequential Monte Carlo methods, {\it Bernoulli} (to appear).

\item
{\sc Doucet}, A., {\sc Montesano}, L. \& {\sc Jasra}, A.~(2006).
Optimal filtering for partially observed point processes using
trans-dimensional sequential Monte Carlo, {\it ICASSP}.

\item
{\sc Doucet,} A., {\sc De Freitas}, J.~F.~G. \& {\sc Gordon},
N. J.~(2001). \emph{Sequential Monte Carlo Methods in Practice}. Springer:
New York.

\item
{\sc Eberle}, A.  \& {\sc Marinelli}, C.~(2011).
Quantitative approximations of evolving probability
measures and sequential Markov chain Monte Carlo methods.
Technical Report, Universitat Bonn.

\item
{\sc Fearnhead}, P.~(2004).
Exact filtering for partially-observed queues. \emph{Statist.
Comp.}, \textbf{14}, 261--266.


\item
{\sc Glynn,} P. W. \& {\sc Meyn} S. P.~(1996).
A Lyapunov bound for solutions of the Poisson equation. {\it Ann. Prob.}, {\bf 24}, 916--931.

\item
{\sc Green}, P. J. (1995). Reversible jump Markov chain Monte Carlo computation
and Bayesian model determination. \emph{Biometrika}, \textbf{82}, 711--732.

\item
{\sc Jasra}, A., {\sc Stephens}, D. A. \& {\sc Holmes}, C. C.~(2007).
On population-based simulation. \emph{Statist. Comp.}, {\bf 17}, 263--279.

\item
{\sc Kantas}, N., {\sc Chopin}, N., {\sc Doucet}, A., {\sc Singh}, S. S.  \& {\sc Maciejowski}, J.~M.~(2011).
On particle methods for parameter estimation in general state-space models.
Technical Report, Imperial College London.

\item
{\sc Liu}, J. S. (2001). \emph{Monte Carlo Strategies in Scientific Computing}.
Springer: New York.

\item
{\sc Pitt}, M. K. \& {\sc Shephard}, N.~(1997).  Filtering via simulation:
Auxiliary particle filters. \emph{J. Amer. Statist. Assoc.}, \textbf{94}, 590--599.

\item
{\sc Roberts}, G. O., {\sc Papaspiliopoulos}, O. \& {\sc Dellaportas}, P. (2004).
Bayesian inference for non-Gaussian Ornstein-Uhlenbeck stochastic
volatility processes. {\it J. R. Statist. Soc. B}, {\bf 66}, 369-393.

\item
{\sc Rousset}, M., \& {\sc Doucet}, A.~(2006).
Discussion of Beskos et al.
\emph{J.~R.~Statist.} \emph{Soc. B}, {\bf 68}, 374--375.

\item
{\sc Rydberg}, T. H., \& {\sc Shephard}, N.~(2000). A modelling
framework for the prices and times of trades made on the New York
Stock exchange. \emph{Non-linear and Non-stationary Signal
Processing}, Ed.~Fitzgerald, W. J., Smith, R. L., Walden, A. T. \&
Young, P. C., 217--246, CUP: Cambridge.

\item
{\sc Shiryaev}, A. (1996). {\it
Probability}, Springer: New York.

\item
{\sc Snyder}, D. L.~(1972).
Filtering and detection for doubly stochastic Poisson processes.
\emph{IEEE Trans. Infor. Th.}, {\bf 18}, 91--102.

\item
{\sc Snyder}, D. L. \& {\sc Miller}, M. I.~(1998). {\it
Random Point Processes in Space and Time},
Springer-Verlag: New York.

\item
{\sc Varini}, E.~(2007).
A Monte Carlo method for filtering a marked doubly stochastic Poisson Process. {\it Stat. Meth. \& Appl.}, \textbf{17}, 183--193.

 \item
 {\sc Whiteley}, N. P., {\sc Johansen}, A.~M.~\& {\sc Godsill}, S. J.~(2011).
 Monte Carlo filtering of piece-wise deterministic processes. 
 {\it J. Comp. Graph. Stat}, {\bf 20}, 119-139.

%
%
%

\end{list}
}

\end{document}